%% file: distinct_degree_factor.tex
\documentclass[12pt]{article}
\usepackage{appendix}
\input{preamble}

\title{A number-theoretic conjecture implying \\ faster algorithms for polynomial factorization \\ and integer factorization}
\author{Chris Umans\thanks{Computing \& Mathematical Sciences Department. Supported by a Simons Investigator Grant.}  \\ Caltech \and Siki Wang\thanks{Computing \& Mathematical Sciences Department. Supported by Umans' Simons Investigator Grant.} \\  Caltech}

\begin{document}  
\setstretch{1.1}
\setcounter{page}{0}
\maketitle
\thispagestyle{empty}

\begin{abstract}
The fastest known algorithm for factoring a degree $n$ univariate polynomial over a finite field $\F_q$ runs in time $O(n^{3/2 + o(1)}\poly\log q)$, and there is a reason to believe that the $3/2$ exponent represents a ``barrier'' inherent in algorithms that employ a so-called baby-steps-giant-steps strategy. In this paper, we propose a new strategy with the potential to overcome the $3/2$ barrier. In doing so we are led to a number-theoretic conjecture, one form of which is that there are sets $S, T$ of cardinality $n^\beta$, consisting of positive integers of magnitude at most $\exp(n^\alpha)$, such that every integer $i \in [n]$ divides $s-t$ for some $s \in S, t \in T$. Achieving $\alpha + \beta \le 1 + o(1)$ is trivial; we show that achieving $\alpha, \beta < 1/2$ (together with an assumption that $S, T$ are structured) implies an improvement to the exponent 3/2 for univariate polynomial factorization. Achieving $\alpha = \beta = 1/3$ is best-possible and would imply an exponent 4/3 algorithm for univariate polynomial factorization. Interestingly, a second consequence would be a reduction of the current-best exponent for deterministic (exponential) algorithms for factoring integers, from $1/5$ to $1/6$. 
\end{abstract}

\newpage

\section{Introduction}
We are interested in the fundamental problem of factoring univariate polynomials over finite fields. In other words, given a polynomial $f(X) \in \F_q[X]$ of degree $n$, given by its coefficients, we wish to produce irreducible polynomials $f_1, f_2, \ldots f_k$ such that $f = f_1f_2\ldots f_k$. This is a classic problem in computer algebra, and polynomial-time algorithms have been known since Berlekamp's algorithm in 1970 \cite{Berlekamp_1970}. Since then, there has been a significant amount of work on algorithms for this central problem (see surveys \cite{von_zur_Gathen_Panario_2001}, \cite{Kaltofen_2003}, \cite{von_zur_Gathen_2006}). A key development due to Kaltofen and Shoup in 1998 \cite{Kaltofen_Shoup_1998} was the first sub-quadratic algorithm for this problem. The current best exponent is $3/2$, achieved by Kedlaya and Umans \cite{KU11} via a new nearly-linear-time algorithm for modular composition. 

These asymptotically fastest algorithms use the framework pioneered by Cantor and Zassenhaus \cite{Cantor_Zassenhaus_1981}, in which there are three main steps to produce the full factorization of $f(X)$ into irreducibles: squarefree factorization, distinct-degree factorization, and equal-degree factorization. The first and third of these steps have algorithms that run in time $\widetilde{O}(n)\poly\log(q)$ \cite{KU11}; distinct-degree factorization is thus the algorithmic bottleneck. This problem is described precisely as follows: 

\begin{problem}[Distinct-Degree Factorization] \label{ddf} Given a monic, squarefree polynomial $f \in \F_q[X]$ of degree $n$, output factors $f_1, f_2,..., f_n \in \F_q[X]$, where $f = f_1f_2 \ccdot f_n$ and each $f_i$ is either 1 or the product of degree-$i$ irreducible polynomials.
\end{problem}
As mentioned, the best-known exponent for Distinct-Degree Factorization is $3/2$ \cite{KU11}. The key algebraic fact used in the algorithm is that the polynomial $g_i(X) = X^{q^i} - X$ is the product of all monic irreducible polynomials of degree dividing $i$. Computing the GCD of $f$ with $g_i(X)$ for $i = 1, 2, \ldots, n$ then gives the distinct-degree factorization. Using fast modular composition, each such GCD can be computed using a nearly linear number of $\F_q$ operations. Since there are $n$ of these GCDs, however, the algorithm remains quadratic. To achieve subquadratic running time, a {\em baby-steps-giant-steps} approach is used in \cite{KU11}, following \cite{Kaltofen_Shoup_1998}, to compute the product  $\prod_{i=1}^ng_i(X) \bmod f(X)$ (and products of contiguous subsets of these factors) in $\widetilde{O}(n^{3/2})$ $\F_q$-operations. This remains the key idea in the best-known algorithms. Indeed, going beyond this algorithmic idea seems to be something of a ``barrier'' -- as \cite{Guo_Narayanan_Umans} give an assortment of problems whose best-known algorithm has exponent $3/2$ together with reductions between them and polynomial factorization. In this sense, all of these problems are stuck at exponent $3/2$ due to the need to go beyond the core baby-steps-giant-steps algorithmic motif. 

In this paper, we observe that instead of computing the GCD of $f(X)$ with $X^{q^i}-X$ for various $i \in [n]$, we can compute the GCD with $X^{q^A} - X$ for $A$ as large as $\exp(n^{1/3})$ (if our target is an exponent-$4/3$ algorithm). Such a high degree polynomial may be computed modulo $f(X)$ by repeated modular composition. Thus a single GCD can ``strip off'' many more than a single distinct-degree factor; it can strip off all those whose associated degree divides $A$, in one shot. Making this new idea work in conjunction with the baby-steps-giant-steps approach to beat exponent $3/2$ leads to the number-theoretic conjecture of the paper's title. The substantive new ideas required by this approach are the subject of the next three sections. Our main theorem (Theorem \ref{thm:main}) gives a randomized algorithm for polynomial factorization with an exponent that depends on the parameters of the number-theoretic conjecture; it has the potential to reach exponent $4/3$ if the conjecture is true for best-possible parameters. 

\paragraph{Outline.} In Section \ref{sec:framework} we describe a general framework for computing the Distinct-Degree Factorization. This leads us to formulate the main number-theoretic conjecture (and variants) in Section \ref{sec:num_theo_conj}. The main algorithmic content is in Section \ref{sec:poly_factoring} which shows how the number-theoretic conjecture can be used to produce the components needed by the general framework. Section \ref{sec:int_factoring} then applies the same ideas to integer factorization. Finally, in Section \ref{sec:ap_conj} we discuss the conjecture, and suggest some questions for future work.

\subsection{Preliminaries}

In this paper we use the notation $\widetilde{O}(\cdot)$ to suppress polylogarithmic factors. 
We use the fact that certain key operations on polynomials in $\F_q[X]$ of degree $n$ can be performed in $\widetilde{O}(n)$ time: multiplication, division with remainder, GCD, and multipoint evaluation (see, e.g., \cite{Von_zur_Gathen_Gerhard_2013}). 

We also use the result from \cite{KU11} that {\em modular composition} can be performed in $\widetilde{O}(n)$ time. This is the operation of computing $f(g(X)) \bmod h(X)$ where $f, g, h$ have degree $n$. Modular composition allows us to raise $X$ to large powers of $q$ modulo a polynomial $f(X)$ very efficiently, as captured in the following proposition:

\begin{proposition}
\label{prop:high-q-power}
Given $h(X) \in \F_q[X]$ of degree $n$, and an integer $a \ge 0$, the polynomial $X^{q^a} \bmod h(X)$ can be computed in time $\widetilde{O}(n)(\log q + \log a)$.  
\end{proposition}

\begin{proof}
We first compute $X^q \mod h(X)$ by repeated squaring, then compose the result with itself modulo $h(X)$ repeatedly to get $X^{q^{2^i}} \bmod h(X)$ for $i = 0, 1, \ldots, \lfloor \log_2 a \rfloor$. After noting that 
\[X^{q^{s+t}} \equiv \left(X^{q^{s}} \bmod h(X) \right)\circ \left(X^{q^{t}} \bmod h(X)\right) \pmod{h(X)}\]
we can produce $X^{q^a} \bmod h(X)$ by composing a subset of the $X^{q^i} \bmod h(X)$ polynomials (corresponding to those $i$ for which $a \bmod 2^{i+1} = 1$).
\end{proof}

A key observation is the following:
\begin{lemma}
\label{lem:difference-polys}
Let $s, t$ be non-negative integers, and let $h(X) \in \F_q[X]$ be monic and squarefree. Then  \[\gcd(X^{q^s} - X^{q^t}, h(X))\]
consists of exactly the irreducible factors of $h$ whose degree divides $|s-t|$, and possibly a scalar factor of $-1$.
\end{lemma}
\begin{proof}
If $s \ge t$, we have that $(X^{q^s} - X^{q^t}) = (X^{q^{s-t}} - X)^{q^t}$;
otherwise $(X^{q^s} - X^{q^t}) = (X - X^{q^{t-s}})^{q^s} $
so the difference or its negation is a power of $X^{q^{|s-t|}} - X$. Thus $X^{q^s} - X^{q^t}$ is a power of a polynomial that is the product of exactly the irreducibles whose degree divides $|s-t|$ (possibly multiplied by $-1$). The lemma follows.
\end{proof}

\section{The framework}
\label{sec:framework}

In this section we describe a general framework for distinct-degree factorization, which captures the previous best algorithms. Our new algorithm is also described within this framework.

The input is a monic, squarefree polynomial of degree $n$. All of the polynomials discussed below are polynomials in the ring $\F_q[X]$.
\begin{theorem} \label{recur-split}
Suppose we have $k = k(n)$ polynomials $g_1(X), g_2(X), \ldots, g_k(X)$, and $\gamma \ge 0$ for which the following hold:
\begin{enumerate}
\item every monic squarefree polynomial of degree $n$ divides $\prod_{i = 1}^k g_i(X)$, and 
\item given a monic squarefree polynomial $h(X)$ with $\deg(h) \le n$ and integers $a \le b$, we can compute the {\em interval polynomial} 
\[\prod_{i=a}^b g_i(X) \bmod h(X)\]
    in time $\widetilde{O}(\deg(h)\cdot n^\gamma)\cdot \poly\log q$, and
\item given a monic squarefree polynomial $h(X)$ of degree at most $n$, and an integer $i$, we can compute the distinct-degree factorization of $\gcd(h, g_i)$ in time $\widetilde{O}(\deg(h) \cdot n^\gamma)\cdot \poly\log q$. 
\end{enumerate}
Then there is a randomized algorithm for computing the distinct degree factorization that runs in time $\widetilde{O}(n^{1 + \gamma}\log(k(n)))\cdot\poly\log q$.
\end{theorem}

The proof implements a strategy we call {\em recursive splitting}. We first describe how the previous algorithm of Kedlaya and Umans fits into this framework with $k(n) = n$ and $\gamma = 1/2$: we define $g_i(X) = X^{q^i} - X$. It is an elementary algebra fact that $g_i(X)$ is the product of all monic, irreducible polynomials whose degree divides $i$. Thus, Property 1 holds. For Property 2, Kedlaya and Umans use the baby-steps-giant-steps approach to show that the interval polynomials can be computed in the required time (details not repeated here). Finally, for Property 3, we use the well-known Divisor Bound: that the number of distinct positive integer divisors of an integer $i$ is at most $\exp(\log i/\log \log i) = i^{o(1)}$. By computing 
\[\gcd(h(X), X^{q^a}-X)\] for $a$ running over these divisors in increasing order, we can compute the distinct-degree factorization of $\gcd(h, g_i)$ in time \[\widetilde{O}(\deg(h)\cdot i^{o(1)})\cdot\poly\log q \le  \widetilde{O}(\deg(h)\cdot n^\gamma)\cdot\poly\log q\] as required. 

\begin{proof}[Proof of Theorem \ref{recur-split}]
We are given $f(X) \in \F_q[X]$ of degree $n$. Our algorithm recursively splits $f(X)$ into factors, ending with each factor being the product of irreducibles of a single degree. Each recursive call comes with an {\em interval} $(a,b)$ (with $1 \le a \le b \le k$) and the invariant that the polynomial to be factored divides the interval polynomial $\prod_{i=a}^b g_i(X)$. Initially, the interval is $(1, k)$ and Property 1 guarantees the invariant for the input polynomial $f(X)$.

The base case of the recursive algorithm is when $a = b$. In this case we return the distinct-degree factorization of $\gcd(f(X), g_a(X))$.

Otherwise, let $c = \lfloor (a + b)/2 \rfloor$ be the midpoint of the interval $(a, b)$. We compute the polynomial
\[f_{\mbox{lower}}(X) = \gcd (f(X), \prod_{i=a}^cg_i(X)).\]
There are three cases. If $f_{\mbox{lower}} = 1$ then $f$ divides $\prod_{i =c+1}^{b} g_i$ and we recursively split $f$ with the interval set to $(c+1, b)$. If $f_{\mbox{lower}} = f$ then we recursively split $f$ with the interval set to $(a, c)$. Otherwise $f_{\mbox{lower}}$ is a non-trival factor of $f$, and we recursively split $f_{\mbox{lower}}$ with the interval set to $(a,c)$ and $f_{\mbox{upper}} = f/f_{\mbox{lower}}$ with the interval set to $(c+1, b)$. In all cases the invariant is maintained.

Let $T(d, r)$ be the running time of this procedure when the input polynomial has degree $d$, and the interval has length $r = b - a +1$. We claim that for all $d \le n$
\[T(d, r) \le (\log (2r))\cdot \widetilde{O}(dn^{\gamma})\cdot\poly\log q.\]
The proof is by induction on $r$. The claim is true by Property 3 when $r = 1$. The recursive step (computing the interval polynomial, and one GCD) costs
\[\widetilde{O}(dn^{\gamma})\cdot\poly\log q\]
by Property 2, plus the time to recursively split $f$ with an interval of length $r/2$ (in the first two cases), or the time to recursively split $f_{\mbox{lower}}$ and $f_{\mbox{upper}}$ each with an interval of length $r/2$ (in the third case). By induction, this additional cost is
\[\log (2(r/2))\cdot \widetilde{O}(dn^{\gamma})\cdot\poly\log q\]
in the first two cases, and 
\[\log (2(r/2))\cdot \widetilde{O}(\deg(f_{\mbox{lower}})n^{\gamma})\cdot\poly\log q + \log (2(r/2))\cdot \widetilde{O}(\deg(f_{\mbox{upper}})n^{\gamma})\cdot\poly\log q\]
in the third case. Since $\deg(f_{\mbox{lower}}) + \deg(f_{\mbox{upper}}) \le d$, this is upper bounded by 
\[\log (2(r/2))\cdot \widetilde{O}(dn^{\gamma})\cdot\poly\log q.\]
Thus in all cases the overall time is at most 
\[\widetilde{O}(dn^{\gamma})\cdot\poly\log q + (\log (2r) - 1)\cdot \widetilde{O}(dn^{\gamma})\cdot\poly\log q \le (\log (2r)\cdot \widetilde{O}(dn^{\gamma})\cdot\poly\log q\]
as claimed\footnote{For ease of exposition, we have not made explicit the function of $d, n, q$ written as $\widetilde{O}(dn^{\gamma})\cdot\poly\log q$. To make the induction formally correct we should fix one explicit function that is simultaneously an upper bound for each occurrence of the expression in this proof, and ensure that that function is convex in $d$.}. 

The running time of the entire algorithm is $T(n, k(n)) \le \widetilde{O}(n^{1 + \gamma}\log(k(n)))\cdot\poly\log q$ as required.
Note that in the final output, each reported factor consists of irreducibles of a given degree, although there may be multiple reported factors with irreducibles of that degree. We can simply group these to produce the final distinct-degree factorization.
\end{proof}

\section{A number-theoretic conjecture}\label{sec:num_theo_conj}

In this section we introduce a family of number-theoretic conjectures, and in the following two sections, we show how, if true, they can be used to produce algorithms for polynomial factorization and integer factorization with exponents that beat the current best-known. We begin with a definition:

\begin{definition}{($n$-divisor property)} 
\label{divisor-cu}
A set $A$ of positive integers satisfies the $n$-divisor property if  for all $i \in \{1, \ldots,  n\}$, there exists $a \in A$ such that $i \mid a$.
\end{definition}

We will be interested in minimizing both the {\em magnitude} of the integers in $A$ and the {\em cardinality} of $A$. Of course the set $\{1,2,\ldots, n\}$ trivially satisfies the $n$-divisor property; a slightly more interesting fact is that the set $\{\lfloor n/2 \rfloor, \ldots, n\}$ of cardinality $\lceil n/2 \rceil$ satisfies the $n$-divisor property. More generally, without any further constraints, it is easy to satisfy the divisor property with sets of cardinality $m$ consisting of integers of magnitude at most $n^{n/m}$. This is close to tight, since the product of primes less than $n$ -- which has magnitude $\exp(n)$ -- must divide the product of the integers in $A$, implying that some element of $A$ must have magnitude $\exp(n)^{1/m}$.

Our main family of conjectures demands that the set $A$ have a difference-set structure. We further parameterize the magnitude and cardinality by $\alpha, \beta$  as follows:
\begin{conjecture} ($(\alpha,\beta)$-Divisor Conjecture)
For infinitely many $n$ there exist subsets $S, T \subseteq \ZZ^{+}$ such that the set of pairwise differences \[A = \{s-t \mid s \in S, t \in T\}\] satisfies the $n$-divisor property, and 
\begin{enumerate}
\item all elements of $S, T$ have magnitude at most $\exp(n^{\alpha})$
\item $|S|, |T| \le n^{\beta}$
\end{enumerate}
\end{conjecture} 

Expressing the argument above in terms of $\alpha, \beta$ yields the constraint $\alpha \ge 1 - 2\beta$; we do not know if $\alpha = 1 - 2\beta$ is achievable. In fact, we do not even know if $\alpha < 1 - \beta$ is achievable (whereas $\alpha = 1 - \beta + o(1)$ is easy by partitioning $[n]$ into $n^{\beta}$ sets of size $n^{1-\beta}$ each, and taking $S$ to be the products of the integers in each of these parts, and $T = \{0\}$).

Achieving $\alpha < 1 - \beta$ is an important challenge that forces one to make use of both sets $S, T$ (unlike the construction above which ``uses'' only one), and we might hope that such a construction would produce structured sets $S, T$. To be useful algorithmically in our intended settings, we need $S$ and $T$ to have the structure of generalized arithmetic progressions:

\begin{conjecture} (Strong $(\alpha,\beta)$-Divisor Conjecture) \label{gap}
For infinitely many $n$ there exist subsets $S, T \subseteq \ZZ^{+}$ such that the set of pairwise differences \[A = \{s-t \mid s \in S, t \in T\}\] satisfies the $n$-divisor property, with 
\begin{enumerate}
\item all elements of $S, T$ have magnitude at most $\exp(n^{\alpha})$
\item $|S|, |T| \le n^{\beta}$
\item There exist $c, c' \le n^{o(1)}$ such that
\[S = \{S_1 + ... + S_c\}, T = \{T_1 + ... +T_{c'}\}.\]
where the $S_i, T_j$ are arithmetic progressions.
\end{enumerate} 
\end{conjecture} 

An interesting fact is that if there are sets $A$ satisfying the $n$-divisor property with $|A| = n^{2\beta}$ and with integers of magnitude $\exp(n^\alpha)$ {\em and $A$ itself is an arithmetic progression}, then the Strong $(\alpha,\beta)$-Divisor Conjecture is true, as captured by the following proposition:
\begin{proposition}
\label{prop:arith-prog-version}
Suppose that for infinitely many $n$, there exists a set $A$ satisfying the $n$-divisor property with $A$ being an arithmetic progression, $|A| \le n^{2\beta}$, and having elements of magnitude at most $\exp(n^\alpha)$. Then the Strong $(\alpha,\beta)$-Divisor Conjecture is true.
\end{proposition}
\begin{proof}
Since $A$ is an arithmetic progression, there exist $b \ge 0$ and $c > 0$ such that \[A = \{b + ic: 0 \le i \le n^{2\beta}\}.\]
We define $S = \{b + i\lceil n^\beta \rceil c: 0 \le i \le n^{\beta}\}$ and $T = \{ic : 0 \le i \le n^{\beta}\}$. Since any nonnegative $i$ less than $n^{2\beta}$ can be written as $i = i_1\lceil n^{\beta} \rceil - i_0$ with $0\le i_1,i_0 \le n^{\beta}$, we find that every element $b+ic \in A$ can be written as $s - t$ with $s = b+ i_1n^{\beta}c \in S$ and $t = i_0 c \in T$. The proposition follows.
\end{proof}
Indeed the set $A = \{1, 2, \ldots, n \}$ trivially satisfies the $n$-divisor property, and via the above proposition the Strong $(\alpha,\beta)$-Divisor Conjecture is true for all $\alpha > 0$ and $\beta = 1/2$. As we will see in the next section, this already implies an exponent $3/2$ algorithm in our framework.

\section{Using the conjecture to factor polynomials in $\F_q[X]$}\label{sec:poly_factoring}

In this section we show that the Strong $(\alpha, \beta)$-Divisor Conjecture yields polynomials $g_i(X)$ meeting the requirements of Theorem \ref{recur-split}. As we remarked, proving the Strong $(\alpha, \beta)$-Divisor Conjecture for any $\alpha < 1 - \beta$ is a challenge (and on the cusp of what is trivially true). We will see in our main theorem (Theorem \ref{thm:randomized-DD}) that proving the conjecture for $\alpha < 1 - \beta$ leads to a subquadratic algorithm, and proving the conjecture for $\alpha, \beta < 1/2$ would lead to a factoring algorithm with exponent strictly less than $3/2$, beating the current best-known algorithm. The best exponent possible in this framework is $4/3$, achieved when $\alpha = \beta = 1/3$ (recall that $\alpha$ must be at least $1 - 2\beta$).

For the remainder of this section, we assume that the Strong $(\alpha, \beta)$-Divisor Conjecture holds for fixed $\alpha, \beta$. The polynomials $g_i(X)$ needed for Theorem \ref{recur-split} are taken to be $X^{q^{s}} - X^{q^t}$ for $s \in S, t \in T$. 

The ordering of the $g_i(X)$ polynomials is important. We associate $g_0, \ldots, g_{|S|\cdot|T|-1}$ with the $S \times T$ as follows: if $s_0, \ldots s_{|S|-1}$ is an enumeration of $S$ and $t_0, \ldots, t_{|T|-1}$ is an enumeration of $T$ then we set 
\begin{equation} 
g_{j+ \ell|S|}(X) = X^{q^{s_j}} - X^{q^{t_\ell}}
\label{eq:gis}
\end{equation}
for $0 \le j \le |S|-1$ and $0 \le \ell \le |T|-1$.
\begin{proposition}
\label{prop:property1}
Suppose that the Strong $(\alpha, \beta)$-Divisor Conjecture holds, and that polynomials $g_i(X)$ are defined as in Equation (\ref{eq:gis}). Then every monic squarefree polynomial of degree $n$ divides $\prod_i g_i(X)$. 
\end{proposition}
\begin{proof}
Let $f(X)$ be a monic squarefree polynomial of degree $n$. We argue that every (monic) irreducible factor of $f(X)$ divides some $g_i(X)$. Let $d \le n$ be the degree of an irreducible factor $h(X)$. By the $n$-divisor property of the set $S-T$, there is some $s \in S, t \in T$ for which $d | (s-t)$ (and hence $d$ divides their absolute value). Let $g_i(X)$ be the polynomial associated with the pair $(s, t)$. Then Lemma \ref{lem:difference-polys} implies that $\gcd(g_i, f)$ contains $h$ as an irreducible factor, and hence $h$ divides $g_i$. 
\end{proof}

\subsection{Computing the interval polynomials}

In this subsection we show how to compute the interval polynomials to satisfy Property 2 of Theorem \ref{recur-split}, assuming the Strong $(\alpha,\beta)$-Divisor Conjecture. Our first lemma shows how to efficiently compute the polynomials $X^{q^u} \bmod h(X)$ for all $u \in S \cup T$.

\begin{lemma}
\label{lem:preprocessing}
Suppose that the Strong $(\alpha, \beta)$-Divisor Conjecture holds, and fix $h(X) \in \F_q[X]$ of degree $d \le n$. There is a procedure running in $\widetilde{O}(d(n^{\alpha+o(1)} + n^{\beta})) \cdot \poly\log q$ time that produces \[X^{q^u} \bmod h(X)\] for all $u \in S \cup T$.
\end{lemma}
\begin{proof}
Recall that $S = S_1 + \cdots + S_c$ for $c \le n^{o(1)}$, where \[S_i = \{a_i + jb_i: j = 0, 1, \ldots, d_i\}\] for integers $a_i \ge 0, b_i > 0$ and $d_i \ge 0$. Observe that we can compute the polynomials $X^{q^{a_i+jb_i}} \bmod h(X)$ for $j = 0,1,\ldots, d_i$ by first computing $X^{q^{a_i}}$ and $X^{q^{b_i}}$, and then composing the first repeatedly with the second, $d_i$ times, with all computations done modulo $h(X)$. Repeating for $i = 1, 2, \ldots, c$ we obtain the polynomials $X^{q^{u_i}}$ for all $u_i \in S_i$, and for all $S_i$. Finally, for each $(u_1, u_2, \ldots, u_c) \in S_1 \times S_2 \times \cdots \times S_c$ we compute 
\[X^{q^{u_1}} \circ X^{q^{u_2}} \circ \cdots \circ X^{q^{u_c}} \pmod{h(X)}.\] 
This gives $X^{q^u} \bmod h(X)$ for every $u \in S$, and we need only to repeat for $T$ to obtain the desired set of polynomials. The overall cost is $O(d)\log q$ to compute $X^q \bmod h(X)$ by repeated squaring, and then \[\sum_{i=1}^c \left ( \lceil \log a_i \rceil + \lceil \log b_i \rceil + (d_i+1) \right ) + |S|\]
modular compositions (using Proposition \ref{prop:high-q-power}), for a total cost of $\widetilde{O}(d(n^{\alpha+o(1)} + n^{\beta})) \cdot \poly\log q$. Here we used the fact that $a_i, b_i \le \exp(n^\alpha)$ and $\prod_i (d_i+1) = |S| \le n^\beta$.
\end{proof}

Now, we show how to efficiently compute a specified interval polynomial. This proof uses as a key step the baby-steps-giant-steps method of Kaltofen and Shoup \cite{Kaltofen_Shoup_1998}.

\begin{lemma}
\label{lem:baby-steps-giant-steps}
Suppose that the Strong $(\alpha, \beta)$-Divisor Conjecture holds, and let the polynomials $g_i(X)$ be as defined in Equation (\ref{eq:gis}). Fix $h(X) \in \F_q[X]$ of degree $d \le n$ and a specified interval $a \le b$. Given the polynomials $X^{q^u} \bmod h(X)$ for each $u \in S \cup T$, there is a procedure that produces the interval polynomial
\[\prod_{i = a}^b g_i(X) \bmod h(X)\]
in time $\widetilde{O}(dn^{\beta}) \cdot \poly\log q$.
\end{lemma}

\begin{proof}
We observe that the integer interval $(a, b)$ is the union of three disjoint sets. Writing $a, b$ in base-$|S|$ we have $a = a_0 + a_1|S|$ and $b = b_0 + b_1|S|$. The three sets are then:
\begin{eqnarray*}
    Q_{\mbox{left}} & = & \{a_0, a_0+1, \ldots, |S|-1\} + a_1|S| \\
    Q_{\mbox{middle}} & = & \{(a_1+1)|S|, \ldots, b_1|S|\} \\
    Q_{\mbox{right}} & = & \{0, 1, \ldots, b_0\} + b_1|S| \\
\end{eqnarray*}
The first set has cardinality $|S|-a_0$, the second has cardinality $(b_1 - a_1)|S|$ (or zero if $a_1 = b_1$), and the third set has cardinality $b_1 + 1$. Since the first and third have size at most $n^\beta$, we can afford to produce the associated polynomials one by one. The set $Q_{\mbox{middle}}$ may have size up to $n^{2\beta}$ so we will use a different strategy to produce the product of its associated polynomials.

Recall from Equation (\ref{eq:gis}) that $g_{j+ \ell|S|}(X) = X^{q^{s_j}}-X^{q^{t_\ell}}$. We produce the polynomials associated with $Q_{\mbox{left}}$ by taking the difference of the polynomials
\[X^{q^{s_{a_0}}}, X^{q^{s_{a_0+1}}}, \ldots, X^{q^{s_{|S|-1}}}\]
with $x^{q^{t_{a_1}}}$ (all modulo $h(X)$). Similarly, we produce the polynomials associated with $Q_{\mbox{right}}$ by taking the difference of the polynomials
\[X^{q^{s_{0}}}, X^{q^{s_{1}}}, \ldots, X^{q^{s_{b_0}}}\]
with $x^{q^{t_{b_1}}}$ (all modulo $h(X)$). 

Since all the polynomials involved so far are given (in the statement of the lemma), producing the  polynomials associated with $Q_{\mbox{left}}$ and  $Q_{\mbox{right}}$ takes $O(|S|)$ operations on degree $d$ polynomials. We take the product of all of them modulo $h(X)$ with an additional $O(|S|)$ operations on degree $d$ polynomials. 

All that remains is to compute the product associated with $Q_{\mbox{middle}}$. For this we use the idea of Kaltofen and Shoup \cite{Kaltofen_Shoup_1998}. Over the ring $R = \F_q[X]/h(X)[Z]$, we form the following polynomial of degree at most $|T|$ 
\[p(Z) = \prod_{j = a_1+1}^{b_1}(Z - X^{q^{t_j}})\]
which entails $O(|T|)$ operations in the ring $R$ since each $X^{q^{t_j}} \bmod h(X)$ is given. We then evaluate $p(Z)$ at $X^{q^{s}}$ for all $s \in S$ using fast multipoint evaluation, which entails $\widetilde{O}(|S| + |T|)$ operations in the ring $R$, since each $X^{q^{s}} \bmod h(X)$ is given. Finally we take the product of these evaluations, again expending at most $O(|S|)$ operations in the ring $R$. Altogether, we obtain the ring element:
\[\prod_{s \in S} \prod_{j = a_1+1}^{b_1}  (X^{q^{s}} - X^{q^{t_j}}) \pmod{h(X)}.\]

We now multiply this polynomial with the polynomials associated with $Q_{\mbox{left}}$ and $Q_{\mbox{right}}$, all modulo $h(X)$, which completes the proof.
\end{proof}

Combining Lemma \ref{lem:preprocessing} and Lemma \ref{lem:baby-steps-giant-steps}, we obtain:
\begin{theorem}
\label{thm:interval-polys}
Suppose that the Strong $(\alpha, \beta)$-Divisor Conjecture holds, and that polynomials $g_i(X)$ are defined as in Equation (\ref{eq:gis}). Then given a monic squarefree polynomial $h(X)$ with $\deg(h) \le n$ and integers $a \le b$, the interval polynomial 
\[\prod_{i=a}^b g_i(X) \bmod h(X)\]
can be produced in time $\widetilde{O}(\deg(h)\cdot n^{\max(\alpha, \beta)+o(1)})\cdot \poly\log q$.
\end{theorem}

Thus the Strong $(\alpha, \beta)$-Divisor Conjecture implies that Property 2 of Theorem \ref{recur-split} holds with $\gamma = \max(\alpha, \beta) + o(1)$.

\subsection{Distinct-degree factorization of $\gcd(h, g_i)$}

Recall that $g_i(X) = X^{q^s} - X^{q^t}$ for some $s \in S, t \in T$, and that $\gcd(h, g_i)$ consists of exactly the irreducibles of $h$ whose degree divides $|s - t|$. Set $P = |s-t|$. 

If $P$ has only few distinct divisors, then we could split off from $h$ the irreducibles of those degrees, one after the other. But it could be that $P$ has {\em many} distinct divisors, if $P$ is the product of small primes. For example, $P$ might be the product of the first $\approx n^\alpha$ primes, and then well-known results on the density of smooth integers tell us that $P$ would have $\Omega(n)$ distinct divisors between $1$ and $n$. So we cannot, in general, afford to try all candidate divisors $d$ and split off the associated distinct-degree factor of $h$ -- this would take time at least $\Omega(n)\deg(h)$ in the worst case, while we are shooting for $\widetilde{O}(n^\alpha\deg(h))$. Instead, we adopt a randomized strategy, described in this section. 

Throughout this section, set $w(X) = \gcd(h(X), g_i(X))$, and let $d \le n$ be its degree; the polynomial $w$ can be computed initially at a cost of $\widetilde{O}(\deg(h)n^\alpha)\cdot\poly\log q$. Let $P$ be the $d$-smooth part of $|s-t|$; i.e., the product of all prime factors of $|s-t|$ that are at most $d$. Note that this can be computed easily by trial division, and that any degree of an irreducible factor of $w$ that divides $|s-t|$ also divides $P$. Let $R$ be the multiset of $P$'s prime factors. Our algorithm is recursive, and in the main step we will either split $w$ non-trivially into $w'(X)$ and $w(X)/w'(X)$, with $\deg(w')$ a substantial fraction of $d$, or reduce the cardinality of $R$, or both. This gives rise below to the following two-variable recurrence (in two different ways, in Lemma \ref{lem:small-R} and Theorem \ref{thm:randomized-DD}).

\begin{lemma}
\label{lem:recurrence}
Let $f, g$ be non-decreasing, non-negative functions, and let $f$ be a convex function on $\RR_{>0}$. Fix $1/2 > \epsilon > 0$ and integers $A, B \geq 1$. Consider the following recurrence:
\begin{equation*}
T(a, b) = \left \{ \begin{array}{ll}g(b) & \mbox{if $a \le A$} \\
f(a) & \mbox{if $b \le B$} \\
f(a)g(b) + \max\limits_{\epsilon a \le a' \le a} \left (T(a', b-1) + T(a-a', b) \right ) & \mbox{otherwise} \end{array} \right .
\end{equation*}
Then for all $a, b > 0$ we have
$T(a, b) \le f(2a)g(b)\log_{1/(1-\epsilon)}(a)b.$
\end{lemma}
\begin{proof}
The proof is by induction on $a$ and $b$. The base cases hold by definition. For the induction step with $a > A$ and $b > B$, we have
\begin{eqnarray*}
T(a, b) & \le & f(a)g(b) + T(a', b-1) + T(a - a', b) \\
& \le & f(a)g(b) + f(2a')g(b-1)\log_{1/(1-\epsilon)}(a')(b-1) + f(2(a-a'))g(b)\log_{1/(1-\epsilon)}(a-a')b \\ 
& \le & f(a)g(b) + f(2a')g(b)\log_{1/(1-\epsilon)}(a)(b-1) + f(2(a-a'))g(b)\log_{1/(1-\epsilon)}(a-a')b \\ 
& \le & f(a)g(b) + (f(2a') + f(2(a-a')))g(b)\log_{1/(1-\epsilon)}(a)b\\ & &  - f(2a')g(b) - 
f(2(a-a'))g(b)\\
& \le & f(a)g(b) + f(2a)g(b)\log_{1/(1-\epsilon)}(a)b - (f(2a') + f(2(a-a')))g(b)\\ 
& \le & f(2a)g(b)\log_{1/(1-\epsilon)}(a)b 
\end{eqnarray*}
where the third-to-last inequality used that $\log_{1/(1-\epsilon)}(a - a') \le \log_{1/(1-\epsilon)}(a) - 1$, the second-to-last inequality used the convexity of $f$, and the last inequality used that $f(2a') + f(2(a-a')) \ge f(a)$ (because either $2a'$ or $2(a - a')$ is at least $a$, their average).
\end{proof}

One easier case occurs when $|R|$ is at most $\poly\log n$. In this case we can compute the distinct degree factorization of $w(X)$ deterministically as follows:

\begin{lemma}
\label{lem:small-R}
Fix an integer $P$ of magnitude $\exp(n^{\alpha})$, and let $R$ be the multiset of $P$'s prime factors, in their multiplicity. Let $w(X)$ be a monic, squarefree polynomial of degree $d \le n$, whose irreducible degrees each divide $P$. There is an algorithm that produces the distinct-degree factorization of $w(X)$ in time $\widetilde{O}(dn^{\alpha}|R|^3)\cdot \poly\log q$.
\end{lemma}
\begin{proof}

We recursively compute the distinct-degree factorization of $w(X)$. In each recursive call, we hold a multiset $R$ with the guarantee that the degree of each irreducible factor of $w$ is a divisor of $\prod_{r \in R} r$. In each step, either the degree $d$ is reduced by at least a $(1-1/|R|)$ factor, or the multiset of prime factors $R$ shrinks. The complete recursive algorithm is as follows:

\begin{algorithm}
    \caption{DDF$(w, R)$}\label{euclid}
\hspace*{\algorithmicindent}
\textbf{Input}: polynomial $w(X)$ of degree $d$, a multiset of prime factors $R$ \\
\hspace*{\algorithmicindent}
\textbf{Guarantee}: the degrees of the irreducible factors of $w$ divide $\prod_{r \in R} r$\\
\hspace*{\algorithmicindent}
\textbf{Output}: the distinct-degree factorization of $w(X)$
\begin{algorithmic}[1]
\If{$d = 1$}\label{line:base1}
    \State output $w(X) = X-a, a \in \F_q$
\ElsIf{$|R| = 1$}
    \State output $w(X)$ 
        \label{line:base2}
    \Else
    \State set $P = \prod_{r \in R} r$
    \label{line:p1-first}
    \State set $w_1 = w$
    \For{$r \in R$}
    \State $w_{r}(X) = \gcd (X^{q^{P/r}} - X, w(X))$ 
    \State $w_1 = w_1/\gcd(w_1, w_r)$ 
    \EndFor     
    \label{line:p1-last}
\EndIf 
\State \textbf{find} $w' \in \{ w_1\} \cup \{w_r : r \in R\}$ of maximum degree
\label{line:max-degree}
\If{$w' = w_1$}
 \State output $w'(X)$
 \State \textbf{return} DDF$(w/w', R)$ (if $w' \ne w$)
 \Else 
 \State \textbf{return} DDF$(w', R \setminus \{r\})$, and DDF$(w/w', R)$ (if $w' \ne w$)
 \EndIf
\end{algorithmic}
\end{algorithm}
There are two base cases handled in lines (\ref{line:base1}- \ref{line:base2}): if $d = 1$ then $w(X) = X-a$ for some $a \in \F_q$; if $|R| = 1$, then it consists of a single prime $p$ and $w(X)$ is the product of irreducible polynomials of degree $p$. In both cases $w$ is trivially its distinct-degree factorization by itself. 

In lines (\ref{line:p1-first} - \ref{line:p1-last}) we set $P = \prod_{r \in R} r$ and the algorithm produces $w_1$ which is exactly the product of irreducibles of $w$ of degree $P$, if there are any. It does so by removing the distinct-degree factors whose associated degree is a proper divisor of $P$. 

A key observation is that in line (\ref{line:max-degree}) the selected polynomial $w'$ has degree $d'$ at least $d/(|R|+1)$.
This is because each irreducible factor of $w$ occurs in at least one of these polynomials (it either has degree $P$ and it occurs in $w_1$, or it has a degree that is a proper divisor of $P$, and it occurs in at least one of the $w_r$ polynomials). The bound on $d'$ follows from a pigeonhole argument.

The recursive calls proceed to compute the distinct-degree factorization of $w/w'$ (unless it equals $1$). In the case that $w' = w_1$, we additionally output $w'$ which is itself the remaining distinct-degree factor. Otherwise, we recursively compute the distinct degree factorization of $w' = w_r$, and we know that the degrees of its irreducible factors divide $P/r$, so we can remove one copy of $r$ from the multiset $R$ when making the recursive call.  

Let $T(d, c)$ denote the running time of this procedure on input $(w, R)$, where $w(X)$ has degree $d$ and the multiset $R$ has cardinality $c$. As noted, if $d = 1$ or $c = 1$ we simply output $w(X)$. The work in lines (\ref{line:p1-first} - \ref{line:p1-last}) costs \[\widetilde{O}(d\cdot n^{\alpha})\cdot \poly\log q\]
for each $r \in R$, which includes the cost of computing $X^{q^{P/r}} - X$ mod $w(X)$ (by repeated squaring and modular compositions), and the subsequent GCDs. The two recursive calls cost $T(d', c-1)$ and $T(d-d', c)$.
By Lemma \ref{lem:recurrence} (with $f(d) = \widetilde{O}(d)$ and $g(c) = c\cdot\widetilde{O}(n^{\alpha})\cdot \poly\log q$, and $\epsilon = 1/(|R|+1)$) we conclude that 
\[T(d, c)\le \widetilde{O}(d\cdot cn^{\alpha})\cdot \poly\log q \cdot \log_{1/(1-\epsilon)}(d)c\]
and after noting that $\log_{1/(1-\epsilon)} d \le (1/\epsilon)\ln d = (|R|+1)\ln d$, the overall running time of $T(d, |R|)$ is as claimed. 
\end{proof}

In general, $|R|$ may be as large as $n^\alpha$, and then we need a different strategy. In this case we repeatedly pick a random subset of $R$ that splits off distinct degree factors whose degree is at least a constant fraction of the degree remaining.

\begin{lemma}
\label{lem:shrink-R-recursive-step}
Fix an integer $P$ of magnitude $\exp(n^{\alpha})$ and let $R$ be the multiset of $P$'s prime factors, in their multiplicity. Let $w(X)$ be a monic, squarefree polynomial of degree $d \le n$, whose irreducible degrees each divide $P$.  Set $p = (1/2)^{1/\log n}$. Pick a random subset $R' \subseteq R$ of cardinality $\lceil p|R| \rceil$ and set $P' = \prod_{r \in R'} r$. Then  
\[\Pr\left [\deg \left (\gcd \left (X^{q^{P'}} - X, w(X)\right )
\right ) \ge d/8 \right ] \ge 1/8\]
provided $|R| \ge 4\log^2 n$.
\end{lemma}
\begin{proof}
Let  $w_1, \ldots, w_k$ be the distinct-degree factors of $w$ and let $d_1, d_2, \ldots, d_k$ be their associated degrees (so each $w_i$ is the product of irreducibles of degree $d_i$). For each $i$ define $R_i$ to be the (multi-)subset of $R$ for which $\prod_{r \in R_i}r = d_i$. Let $I_i$ be the indicator random variable for the event that $R_i \subseteq R'$ and set $I = \sum_{i=1}^k \deg(w_i) I_i$. Set $c = |R|$ and $c' = |R'|$. We have 
\[\Pr[R_i \subseteq R'] = \frac{{{c - |R_i|} \choose {c' - |R_i|}}}{{c \choose {c'}}} \ge \left (\frac{c'-|R_i|}{c}\right )^{|R_i|} \ge \left (p- \frac{|R_i|}{c}\right )^{|R_i|} \ge p^{|R_i|}  - \frac{|R_i|^2}{c} \ge 1/2 - 1/4\]
by our choice of $p$, and the fact that $|R_i| \le \log n$ and $c \ge 4 \log^2 n$.
Thus by linearity of expectation, $\mathbb{E}[I] \ge d/4$. By an averaging argument, $\Pr[I < d/8] < 7/8$, and the lemma follows.
\end{proof}

Our final algorithm produces the complete distinct-degree factorization of $w(X)$, invoking Lemma \ref{lem:small-R} when $|R|$  shrinks below $4\log^2 n$. 

\begin{theorem}
\label{thm:randomized-DD}
Fix an integer $P$ of magnitude $\exp(n^{\alpha})$. Let $w(X)$ be a monic, squarefree polynomial of degree $d \le n$, whose irreducible degrees each divide $P$. 
There is a randomized algorithm that produces the distinct-degree factorization of $w(X)$ in time $\widetilde{O}(dn^{\alpha})\poly\log q$.
\end{theorem}

\begin{proof}
Let $R$ be the multiset of $P$'s prime factors, in their multiplicity, and let $p$ be as defined in Lemma \ref{lem:shrink-R-recursive-step}. We recursively compute the distinct-degree factorization of $w(X)$. In each recursive call, we hold a multiset $R$ with the guarantee that the degree of each irreducible factor of $w$ is a divisor of $\prod_{r \in R} r$. In each step, either the degree $d$ is reduced by at least a constant factor, or the multiset of prime factors $R$ shrinks by a $p$ factor. 

If $d = 1$ then we return $w$ and we are done. If $|R| \le 4\log^2 n$ then we invoke Lemma \ref{lem:small-R} and we are done. Otherwise, we pick $R'\subseteq R$ randomly and define $P' = \prod_{r \in R'} r$ as in Lemma \ref{lem:shrink-R-recursive-step}. Compute $w'(X) = \gcd(X^{q^{P'}} - X, w(X))$ and set $d' = \deg(w')$. By Lemma \ref{lem:shrink-R-recursive-step}, with probability at least $1/8$, we have $d' \ge d/8$. We check this condition and repeat if it fails to hold. Repeating $O(\log n)$ times drives the probability of failure down to $1/\poly(n)$, so that by a union bound, the overall algorithm, which invokes this randomized step fewer than $\poly(n)$ times, succeeds with at least constant probability.

At this point we recursively compute the distinct-degree factorization of $w'(X)$  with $R$ replaced by $R'$, and the distinct-degree factorization of $w(X)/w'(X)$ with $R$ unchanged (we don't make the second recursive call if $w = w'$).

Let $T(d, c)$ denote the running time of this procedure on input $(w, R)$, where $w(X)$ has degree $d$ and $c = \lceil \log_{1/p} |R| \rceil$ (note that in the present analysis $c$ is $O(\log |R|)$, instead of $|R|$ as it was in the proof of Lemma \ref{lem:small-R}). If $d = 1$ we simply return $w$; if $|R| \le 4\log^2 n$, we invoke Lemma \ref{lem:small-R} at cost $f(d) = \widetilde{O}(d) \cdot \widetilde{O}(n^{\alpha}(4\log^2 n)^3)\cdot\poly\log q$. 

The main step costs $O(c) \le O(n^\alpha)$ to select $R' \subseteq R$ and $\widetilde{O}(d\cdot n^{\alpha})\cdot \poly\log q$
to account for computing $X^{q^{P'}} - X$ mod $w(X)$ (by repeated squaring and modular compositions), and the subsequent GCD, up to $O(\log n)$ times. The two recursive calls cost $T(d', c-1)$ and $T(d-d', c)$ (or just $T(d', c-1)$ if $d' = d$).
By Lemma \ref{lem:recurrence} (with $f(d) = \widetilde{O}(d) \cdot \widetilde{O}(n^{\alpha})\cdot \poly\log q$,  $g(c) = c$, and $\epsilon = 1/8$) we conclude that 
\[T(d, c)\le \widetilde{O}(d\cdot cn^{\alpha})\cdot \poly\log q \cdot \log_{8/7}(d)c\]
and thus the overall running time is at most $T(d, \log_{1/p}|R|)$, which is as claimed, after noting that $\log_{1/p} |R| \le O(\log n)\ln |R|$.
\end{proof}

Thus the Strong $(\alpha, \beta)$-Divisor Conjecture implies that Property 3 of Theorem \ref{recur-split} holds with $\gamma = \alpha$. We note also that by taking $P$ to be the product of primes less than $n^{\alpha}$ the theorem gives an unconditional, randomized, distinct-degree factorization algorithm that succeeds on polynomials whose irreducible degrees are $n^\alpha$-smooth, which may be of independent interest.

Altogether, by plugging Proposition \ref{prop:property1}, Theorem \ref{thm:interval-polys}, and Theorem \ref{thm:randomized-DD}, into Theorem \ref{recur-split} we obtain:

\begin{theorem}[main]
\label{thm:main}
Suppose that the Strong $(\alpha, \beta)$-Divisor Conjecture holds. Then there is a randomized algorithm that computes the distinct-degree factorization of a monic squarefree polynomial of degree $n$ in time 
$\widetilde{O}(n^{1 + \max(\alpha, \beta)+o(1)})\cdot \poly\log q.$
\end{theorem}

\section{Using the conjecture to factor integers}\label{sec:int_factoring}

In this section, we show that essentially the same ideas as used in Section \ref{sec:poly_factoring} can be used to give deterministic factoring algorithms for integers. In this regime, the state-of-the-art algorithms are exponential algorithms with running time of the form $\widetilde{O}(N^\gamma)$ for $0 < \gamma < 1$, where $N$ is the integer to be factored. This contrasts with randomized integer factoring which is heuristically $\exp(\widetilde{O}(\log^\gamma N))$ time (e.g. the Number Field Sieve). The trivial trial division method has running time $\widetilde{O}(N^{1/2})$, and classic results of Pollard \cite{Pollard_1974} and Strassen \cite{Strassen_1977} give $\widetilde{O}(N^{1/4})$. Recent work has improved the exponent from $1/4$ to $1/5$ using additional ideas \cite{Harvey_2021}. If a version of our strong $(\alpha, \beta)$-Divisor Conjecture is true for $\alpha, \beta = 1/3$, we show that this improves the exponent further to $1/6$, and would give a deterministic algorithm with the best known running time. 

In order to apply the conjecture, we need to assume an additional property. It is likely that if the conjecture is true, the construction of the structured sets $S, T$ would be done in such a way that the different divisors in the $n$-Divisor property would be allocated to the various $s-t$ differences in a prescribed or ``known'' fashion. In this case it is plausible to expect an efficient deterministic function which lists the prime factors of a given $s-t$ difference, in time $\widetilde{O}(n^{\alpha})$. For ease of exposition, we give a name to the version of the conjecture which presumes the existence of such a function, as follows:

\begin{conjecture} (Strong Prefactored $(\alpha,\beta)$-Divisor Conjecture)
\label{conj:strong-prefactored}
For infinitely many $n$, there exists subsets $S, T \subseteq \ZZ^{+}$ that satisfy the requirements of Conjecture \ref{gap}. Moreover, for every integer $s-t$ with $s \in S, t \in T$, its prime factorization can be output in deterministic time $\widetilde{O}(n^{\alpha})$. 
\end{conjecture}

As we will see, if the Strong Prefactored $(\alpha,\beta)$-Divisor Conjecture holds for $\alpha, \beta < 2/5$ then there exists a determinstic algorithm to factor integers that beats the current-best exponent of $1/5$; if the conjecture holds for $\alpha, \beta = 1/3$ (which is best possible), then the exponent is $1/6$.

\subsection{The framework}

For the remainder of this section, we will use the fact that common operations on integers $n$, such as addition, multiplication, division with remainder, and computing GCD run in time $\widetilde{O}(\log n)$. As a standard reference, see \cite{Von_zur_Gathen_Gerhard_2013}.

We now describe a general framework for integer factorization. This is the analogue of Theorem \ref{recur-split}, except for the added assumption that for every $A_i$, its prime factors are efficient to produce. 

\begin{theorem}\label{recur-split-int}
Let $n = \lfloor N^{1/2} \rfloor$. Suppose we have $k = k(n)$ integers $A_1, A_2, \ldots, A_k$, and $\gamma \ge 0$ for which the following hold:
\begin{enumerate}
\item every integer $\leq n$ divides $\prod_{i = 1}^k A_i$, and 
\item given $d \le N$ and integers $a \le b$, we can compute the {\em interval product} 
\[\prod_{i=a}^b A_i \bmod d\]
    in time $\widetilde{O}(\log d \cdot n^\gamma)$, and
\item given an integer $i$, there exists a function $f$ that returns the prime factors of $A_i$ in time $\widetilde{O}(n^{\gamma})$.
\end{enumerate}
Then there is a deterministic algorithm to compute the prime factorization of $N$ in time $\widetilde{O}(N^{\gamma/2}\log(k(n)))$.   
\end{theorem}

\begin{proof}
Initially, we set $N_0 = \gcd(\prod_{i=1}^{k} A_i, N)$. This can be computed within the claimed time bound by Property 2 and a GCD. 

To produce $N$'s prime factorization, we use recursive splitting to eventually obtain at most $k' = O(\lceil \log N_0 \rceil)$ pairs of integers $(A_i, d_i)$, where $A_i \in \{A_1,...,A_k\}$, $d_i | A_i$, and $\prod_{i = 1}^{k'} d_i = N_0$. Note that this is not yet the final factorization of $N$; however it is true that for every prime $p\le n$ that divides $N$, it occurs as a factor of some $d_i$, by Property 1. 

We now obtain the list of primes that divide $N_0$. To do this, for each $i$ we obtain prime factors of $A_i$ by Property 3, and check for each prime $p$, whether $p \mid d_i$. This step takes $\widetilde{O}(n^{\alpha})$ by the fact that $A_i$ has magnitude $\exp(n^{\alpha})$, and there are at most $O(\lceil \log N_0 \rceil)$ pairs. 

Once we have this set of primes, we trial divide each with $N$ for at most $O(\lceil \log N \rceil)$ rounds to determine its multiplicity in $N$. This step cost at most $\widetilde{O}(\lceil \log N \rceil^2)$. Combining this multiset of primes with the (at most 1) prime factor of $N$ which is greater than $n$ gives the full factorization of $N$. 

Now, all that remains is to prove we can obtain the required set $\{(A_i, d_i)\}$ using recursive splitting. The argument then proceeds similarly to that of Theorem \ref{recur-split}. During each recursive call, we keep track of the interval $(a,b)$ while maintaining the invariant that its input $d$ divides $\prod_{i=a}^b A_i$. Initially, the interval is $(1, k)$ and the initial input $d = N_0$ satisfies the invariant by definition. 

The base case of the recursive algorithm is when $a = b$. In this case, we just return the pair $(A_a, \gcd( A_a,d))$. Otherwise, let $c = \lfloor (a + b)/2 \rfloor$ be the midpoint. We compute the integer 
\[d_{\mbox{lower}}=\gcd\left (\prod_{i = a}^c A_i, d\right ).\]

Just like the proof of Theorem \ref{recur-split}, if $d_{\mbox{lower}} = 1$ then $d$ divides $\prod_{i =c+1}^{b} A_i$ and we recursively split $d$ with the interval set to $(c+1, b)$. If $d_{\mbox{lower}} = d$ then we recursively split $d$ with the interval set to $(a, c)$. Otherwise $d_{\mbox{lower}}$ is a non-trivial factor of $d$, and we recursively split $d_{\mbox{lower}}$ with the interval set to $(a,c)$ and $d_{\mbox{upper}} = d/d_{\mbox{lower}}$ with the interval set to $(c+1, b)$. 


We now analyze the running time. Let $T(d, r)$ be the running time of this procedure when the input integer is $d \leq N_0$, and the interval has length $r = b - a +1$. We claim that for all $d$,
\[T(d, r) \le \log (2r)\cdot \widetilde{O}(\log d \cdot n^{\gamma}).\]
The proof is by induction on $r$. The claim is true when $r = 1$. The recursive step (computing the product $d_{\mbox{lower}}$, and one GCD) costs $\widetilde{O}(\log d \cdot n^{\gamma})$
by Property 2, plus the time to recursively split $d$ with an interval of length $r/2$ (in the first two cases), or the time to recursively split $d_{\mbox{lower}}$ and $d_{\mbox{upper}}$ each with an interval of length $r/2$ (in the third case). By induction, this additional cost is
\[\log (2 \cdot r/2)\cdot \widetilde{O}(\log d \cdot n^{\gamma})\] in the first two cases, and 
\[\log (2 \cdot r/2)\cdot \widetilde{O}(\log(d_{\mbox{lower}})n^{\gamma})+ \log (2 \cdot r/2)\cdot \widetilde{O}(\log(d_{\mbox{upper}})n^{\gamma})\]
in the third case. Since $d_{\mbox{lower}} \cdot d_{\mbox{upper}} = d$, in all cases the overall time is upperbounded by 
\[\widetilde{O}(\log d \cdot n^{\gamma}) + (\log 2r - 1)\widetilde{O}(\log d \cdot n^{\gamma}) \leq \log (2r)\cdot \widetilde{O}(\log d \cdot n^{\gamma})\]
which satisfies the inductive hypothesis.
\end{proof}

\subsection{Computing $\prod_{i = a}^b A_i \bmod d$ fast}
In this section, we show that the the Strong $(\alpha, \beta)$-Divisor Conjecture yields interval products $\prod_{i = a}^b A_i$ that can be computed fast, which satisfies Property 2 of Theorem \ref{recur-split-int}. To utilize the structure of the sets $S$ and $T$ in the the Strong $(\alpha, \beta)$-Divisor Conjecture, we use the ordering (described in the previous section) on  $A = \{s-t \mid s\in S, t \in T\}$: if $s_0, \ldots s_{|S|-1}$ is an enumeration of $S$ and $t_0, \ldots, t_{|T|-1}$ is an enumeration of $T$, we set 
\begin{equation} 
A_{j+ \ell|S|} = s_j - t_{\ell}
\end{equation}
for $0 \le j \le |S|-1$ and $0 \le \ell \le |T|-1$.\\

Similar to polynomial factorization, we now want to show that it is efficient to obtain the set of integers $S \cup T$ modulo $d$, which gives us a way to compute any interval $\prod_{i = a}^b A_i$ within the required budget. 

\begin{lemma}\label{lem:int-preprocessing}
Suppose that the Strong $(\alpha, \beta)$-Divisor Conjecture holds, and fix a positive integer $d \le n$. There is a procedure running in $\widetilde{O}((\log d)(n^{\alpha+o(1)} + n^{\beta}))$ time that produces $u \bmod d$ for all $u \in S \cup T$.
\end{lemma}

\begin{proof}
See Appendix.
\end{proof}

\begin{lemma}\label{lem:int_bsgs}
Suppose that the Strong $(\alpha, \beta)$-Divisor Conjecture holds. Fix a positive integer $d \le n$ and a specified interval $a \le b$. Given the integers $u \bmod d$ for all $u \in S \cup T$, there is a procedure that produces
\[\prod_{i = a}^b A_i \bmod d\]
in time $\widetilde{O}((\log d) n^{\beta})$.
\end{lemma}

\begin{proof}
See Appendix.
\end{proof}

Altogether, if we assume the Strong Prefactored $(\alpha, \beta)$-Divisor Conjecture, combining Lemma \ref{lem:int-preprocessing}, Lemma \ref{lem:int_bsgs} with Theorem \ref{recur-split-int}, for $\gamma = \max(\alpha, \beta) + o(1)$, we obtain an unconditional, deterministic algorithm for integer factorization.

\begin{theorem}[main -- integer factoring]
Suppose that the Strong Prefactored $(\alpha, \beta)$-Divisor Conjecture holds. Then there is a deterministic algorithm that computes the prime factorization of an integer $N$ in time 
$\widetilde{O}(N^{\max(\alpha, \beta)/2+o(1)}).$
\end{theorem}

\section{The arithmetic progression version of the conjecture}\label{sec:ap_conj}

The formally strongest version of the Strong $(\alpha, \beta)$-Divisor Conjecture is that there are {\em arithmetic progressions} $A$ of cardinality $n^{2\beta}$, whose elements have magnitude $\exp(n^{\alpha})$, that satisfy the $n$-divisor property. We call this the ``Arithmetic Progression Version'' of the conjecture. Proposition \ref{prop:arith-prog-version} shows that the Arithmetic Progression Version implies the Strong $(\alpha, \beta)$-Divisor Conjecture. Thus, if one is to disprove a version of the conjecture, the easiest version to attack is the arithmetic progression one. Here we make some observations that may be useful in studying this version. 

First, recall that such an arithmetic progression $A = \{b + ic: 1\le i \le n^{2\beta}\}$ satisfying the $n$-Divisor Property yields an improved exponent for polynomial factorization if $\alpha, \beta < 1/2$. One observation is that $b \ne 0$ is necessary for non-triviality. For if $b = 0$ then we have that every prime $p< n$ divides $ic$ for some $i \le n^{2\beta}$, and this requires $c$ to be the product of all primes between $n^{2\beta}$ and $n$, which is a number of magnitude $\exp(n - o(1))$ when $\beta < 1/2$. It is also necessary that $(b+c) \ge \exp(n^{1-2\beta})$ as we have argued in Section \ref{sec:num_theo_conj}. Although it seems unlikely, we cannot rule out the possibility that the Arithmetic Progression Version may hold with $b \le \exp(O(n^{1 - 2\beta}))$ and $c = 1$!

Certain {\em subsets} of integers in $[n]$ are easy to handle: for example, by choosing $c$ to be the least common multiple of all positive integers of magnitude at most $n^\alpha$, and $b=0$, it is easy to see that the arithmetic progression $A = \{b + ic: 1\le i \le n^{2\beta}\}$  contains multiples of each integer that can be written $xy$ where $x \le n^{2\beta}$ and $y \le n^{\alpha}$. On the other hand, some subsets of integers seem difficult. Two challenging sets of integers are (1) the primes in $[n]$, and (2) composites $xy$ where $x, y$ are in the range $[n^{1/2 - \epsilon}, n^{1/2}]$ for a constant $\epsilon > 0$.

A useful fact is that different arithmetic progressions containing multiples of different subsets of integers can be combined into a single generalized arithmetic progression:

\begin{lemma}
Let $A_1 = \{b_1 + ic_1: 1\le i \le \ell_1\}$  and $A_2 = \{b_2 + ic_2: 1\le i \le \ell_2\}$ be two arithmetic progressions. Then the  generalized arithmetic progression
\[A = \{c_1b_2 + i(c_1c_2) + j(c_2b_1 - c_1b_2): 1\le i \le \max(\ell_1, \ell_2), \;\; 0 \le j \le 1\}\]
contains multiples of $x$ if $A_1$ or $A_2$ does.
\end{lemma}
\begin{proof}
Suppose $x$ divides $b_1 + ic_1$ for some $i$. Then the $(i, 1)$ term in $A$, which equals $c_2b_1 + i(c_1c_2) = c_2(b_1 + ic_1)$ is also divisible by $x$. Similarly, if $x$ divides $b_2 + ic_2$ for some $i$, then the $(i, 0)$ term in $A$, which equals $c_1b_2 + i(c_1c_2) = c_1(b_2 + ic_2)$ is divisible by $x$.
\end{proof}

One can combine a constant number of arithmetic progressions, of length $n^{2\beta}$ and with coefficient magnitude $\exp(n^\alpha)$, in this fashion, and the magnitude of the coefficients increases from $\exp(n^{\alpha})$ to only $\exp(O(n^\alpha))$, and the size of the arithmetic progression increases from $n^{2\beta}$ to only $O(n^{2\beta})$. It is not hard to extend Proposition \ref{prop:arith-prog-version} to generalized arithmetic progressions of this form; i.e., ones where the generalized arithmetic progression is the sum of one ``long'' arithmetic progression and a constant number of others of constant length. Thus solving the problem separately for constantly many subsets of the integers in $[n]$ is a potential route to proving the Strong $(\alpha, \beta)$-Divisor Conjecture. 

As mentioned, the subset of primes in $[n]$ is a particularly challenging case. Here we have a interesting characterization:
\begin{lemma}\label{lem:prime_char}
Let $U$ be the product of primes in $[n]$. Given a factorization $U = U_1U_2 \cdots U_\ell$, define $T_i$ to be the multiplicative inverse of $U/U_i \pmod{U_i}$, and define $V_i = (U/U_i)T_i$. There is an arithmetic progression 
$A = \{b + ic: 1\le i \le \ell\}$ containing multiples of each prime in $[n]$, iff there is a factorization $U=U_1U_2 \cdots U_\ell$ and integers $b,c$ for which 
\begin{equation}
\label{eq:characterization-lemma}
c \cdot \left (\sum_{i = 1}^{\ell} iV_i\right ) + b \equiv 0 \pmod{U}.
\end{equation}
\end{lemma}

\begin{proof}
One direction is easy: if $U = U_1U_2 \cdots U_\ell$ is a factorization for which there exist integers $b,c$ such that Equation (\ref{eq:characterization-lemma}) holds, 
then by considering the equation modulo $U_i$ we see that $ci + b \equiv 0 \pmod{U_i}$ and thus each prime factor of $U_i$ divides the $i$-th term in the arithmetic progression $A$.

In the other direction, we start with the assumption that each prime in $[n]$ divides some term in the arithmetic progression $A = \{b + ic: 1\le i \le \ell\}$. Let $U = U_1U_2 \cdots U_\ell$ be a consistent factorization in the sense that $p \mid U_i$ implies $p$ divides the $i$-th term in $A$. We claim that every $p \in [n]$ divides 
\begin{equation}
\label{eq:characterization}
c \cdot \left (\sum_{i = 1}^{\ell} iV_i\right ) + b
\end{equation}
and hence it equals $0$ modulo $U$ by the Chinese Remainder Theorem. As above, we see that modulo $U_i$, expression 
(\ref{eq:characterization}) is equivalent to $ci+b$, and hence it is divisible by every $p$ that divides $U_i$, by our choice of the factorization to be consistent. Since every prime $p \in [n]$ divides some $U_i$, this argument gives us that every such $p$ divides expression 
(\ref{eq:characterization}). Thus it equals 0 modulo $U$ as claimed. 
\end{proof}

Note that the parenthesized quantity in Equation (\ref{eq:characterization-lemma}) is, in general, of magnitude comparable to $U$, which is $\exp(n)$, while we wish for $b, c \le \exp(n^\alpha)$. Thus we are asking for a ``small'' affine transformation of a ``large'' quantity to equal a multiple of $U$. The design freedom is in the choice of the factorization of $U$, which determines the ``large'' quantity. If such a factorization is not possible, then this would disprove the Arithmetic Progression Version of the Strong $(\alpha, \beta)$-Divisor Conjecture.


Finally we note that if the Arithmetic Progression Version of the conjecture is true, then it cannot be true ``too generically'' in the sense expressed in the following lemma:
\begin{lemma}
Suppose for all positive integers $A_1, A_2, \ldots, A_\ell$ with $A_i \le M$, there exist positive integers $A_1', A_2', \ldots, A'_{\ell'}$ in arithmetic progression, with $A'_j \le M'$, such that for all $i$, there exists $j$ such that $A_i |A_j'$. Then either $\ell' \ge M^{\Omega(1)}$ or $M' \ge M^{\log\log M}$.
\end{lemma}
This rules out the hope that one could (for example) argue the existence of an arithmetic progression starting only from an arbitrary partition of $[n]$ into $n^{1-2\beta}$ sets, setting $A_i$ to be the product of the integers in the $i$-th part, and ``embedding'' these factors into an arithmetic progression with parameters that do not degrade by much. Interestingly, this lemma leaves open the possibility that the parameter $M$ could degrade only mildly, which would still be suitable for our application. Is the bound in the lemma tight? If so, we would achieve exponent $4/3$ via Theorem \ref{thm:main}; we expect however that it is not tight, and that one would need to use more specific number-theoretic properties of the partition that would make it amenable to appearing as factors of an arithmetic progression.  

\begin{proof}
The proof is by counting two ways. The number of possible $A_i$ sequences is at least $\exp(\ell \log_2 M)$. If the $A_j'$ form an arithmetic progression $\{b + jc\}$ then we can specify the original $A_i$ sequence by writing down $b, c$, and for each $i$ the associated $j$ for which $A_i | A_j'$ together with a specification of that divisor of $A_j'$. Since by the Divisor Bound there are at most $\exp(\log x / \log \log x)$ distinct divisors of an integer $x$, an upper bound on the number of possible $A_i$ sequences is in total, at most: 
\[\exp\left ( 2\log_2 M' + \ell \left (\log_2 \ell' + \frac{\log M'}{\log \log M'} \right )\right ),\]
(which is exponential in the length of the description we specified above). 
If $\ell' \le M^{o(1)}$ and $M' \le M^{\log\log M}$ then this quantity is smaller than $\exp(\ell \log_2 M)$, a contradiction.
\end{proof}

\section{Conclusions and open problems}\label{sec:conclusion}

It is certainly possible that the Arithmetic Progression Version of the Strong $(\alpha, \beta)$-Divisor Conjecture is false, and the same could be true for any or all of the other variants. Even though some variants of the conjecture may turn out to be false, we note that there is still room for the ideas and algorithmic machinery in this paper to lead to improved algorithms for factoring polynomials and integers. This is because certain generalizations of the main conjectures are possible, and these generalizations could be true even if the versions we have discussed here are false. For example, one can replace the requirement that $S, T$ are generalized arithmetic progressions with the more general requirement that they have short {\em addition sequences}. This generalization of the classic {\em addition chains} was introduced by Knuth in \cite{Knuth_1997} ($\S 4.6.3$ problem $32$) and by Downey, Leong and Sethi in \cite{Downey_Leong_Sethi_1981}, and some nontrivial results are known in \cite{Downey_Leong_Sethi_1981} and \cite{Dobkin_Lipton_1980}. Addition sequences are sufficient because Lemma \ref{lem:preprocessing} remains true if $S, T$ have short addition sequences, with generalized arithmetic progressions being just one example of a suitably short addition sequence. 

Beyond directly proving or refuting the main conjectures, we highlight the case mentioned in Section \ref{sec:ap_conj}, of finding a length-$n^{2\beta}$ arithmetic progression consisting of integers of magnitude $\exp(n^\alpha)$ that contains multiples of the composites $xy$ with $x, y$ in the range $[n^{1/2 - \epsilon}, n^{1/2}]$, for a constant $\epsilon > 0$. Perhaps one can find a characterization similar to Lemma \ref{lem:prime_char} for this subset of integers. Alternatively, can one give a reduction to the ``prime case'' discussed in Lemma \ref{lem:prime_char}? It would be interesting to establish that the prime case is the ``hardest'' case formally. 

On the algorithmic side, we wonder whether a deterministic version of Theorem \ref{thm:randomized-DD} is possible. This would make our overall algorithm for distinct-degree factorization deterministic; the other algorithms (including the best-known exponent $3/2$ one) for this problem are deterministic. We still wouldn't obtain fast deterministic algorithms for polynomial factorization, because the known equal-degree factorization algorithms are randomized, but derandomizing Theorem \ref{thm:randomized-DD} is still an interesting challenge by itself.

Finally, it remains open to find a nearly-linear time algorithm for univariate polynomial factorization over a finite field, with or without the number-theoretic conjecture, which can prove exponent 4/3 at best. 

\bibliographystyle{alpha}
\bibliography{ref}

\begin{appendices}
\section{Deferred proofs}

\subsection{Proof of Lemma \ref{lem:int-preprocessing}} Recall that $S = S_1 + \cdots + S_c$ for $c \le n^{o(1)}$, where \[S_i = \{a_i + jb_i: j = 0, 1, \ldots, d_i\}\] for integers $a_i \ge 0, b_i > 0$ and $d_i \ge 0$. We can compute $a_i+jb_i \bmod d$ for $j = 0,1,\ldots, d_i$ by first computing $a_i \bmod d$, $b_j \bmod d$, and then add the second to the first $d_i$ times modulo $d$. Repeating for $i = 1, 2, \ldots, c$ we obtain $u_i \bmod d$ for all $u_i \in S_i$, and for all $S_i$. Finally, to obtain all $u \in S$ we do \[u \bmod d= (u_1 \bmod d) + \ldots + (u_c \bmod d) \mod d\]  and we repeat the same procedure for $T$. The overall cost is \[\sum_{i=1}^c \left ( \lceil \log a_i \rceil + \lceil \log b_i \rceil + (d_i+1) \right ) + |S|\]
modulo $d$. We obtain a total cost of $\widetilde{O}((\log d)(n^{\alpha+o(1)} + n^{\beta}))$. Here we used the fact that $a_i, b_i \le \exp(n^\alpha)$ and $\prod_i (d_i+1) = |S| \le n^\beta$.\\

\subsection{Proof of Lemma \ref{lem:int_bsgs}} Writing $a, b$ in base-$|S|$ we have $a = a_0 + a_1|S|$ and $b = b_0 + b_1|S|$. we divide integers $\{A_a,...,A_b\}$ into three sets with the following indices
\begin{eqnarray*}
    Q_{\mbox{left}} & = & \{a_0, a_0+1, \ldots, |S|-1\} + a_1|S| \\
    Q_{\mbox{middle}} & = & \{(a_1+1)|S|, \ldots, b_1|S|\} \\
    Q_{\mbox{right}} & = & \{0, 1, \ldots, b_0\} + b_1|S| \\
\end{eqnarray*}
The first and third set have cardinality at most $|S|$. We can produce the product of all the integers modulo $d$ using Lemma \ref{lem:int-preprocessing}. This takes time $\widetilde{O}((\log d )n^{\beta})$. All that remains is to compute the product modulo $d$ associated with $Q_{\mbox{middle}}$. We again use fast multipoint evaluation, this time forming a polynomial of degree at most $|T|$ over $\ZZ/d\ZZ$:
\[p(Z) = \prod_{j = a_1+1}^{b_1}(Z - t_j).\]
 We then evaluate $p(Z)$ at $s$ for all $s \in S$ using fast multipoint evaluation, which entails $\widetilde{O}(|S| + |T|)$ operations in the ring $\ZZ/d\ZZ$. Then, we take the product of these evaluations using $O(|S|)$ operations. Altogether, we obtain the ring element:
\[\prod_{s \in S} \prod_{j = a_1+1}^{b_1}  (s - t_j) \bmod d.\]

Multiply this product with those associated with $Q_{\mbox{left}}$ and $Q_{\mbox{right}}$ will produce $\prod_{i = a}^b A_i \bmod d$, which completes the proof. 
\end{appendices}

\end{document}


\begin{lemma}
Suppose the arithmetic progression $A = \{b + ic: 1\le i \le n^{2\beta\}$ satisfies the $n$-Divisor Property. Then $c \ge \exp(n^{1 - 2\beta})$.
\end{lemma}
\begin{proof}
Consider the primes in the range $[n^{2\beta}+1, \ldots, n^{2\beta} + \widetilde{O}(n^{1 - 2\beta})]$ where the interval is chosen to contain at least $2\beta n^{1 - 2\beta}$ distinct primes. If every one of these primes divides $c$, then we are done. Otherwise, let $q$ be a prime that does not divide $c$ and let $p_1, p_2, \ldots p_k$ be an enumeration of primes smaller than $n/q$. Note that by the $n$-Divisor Property, for each $1 \le j \le k$ there is some $i_j$ for which $pq_i$ divides $b+i_jc$. Then for any {\em two} $j, j'$ we have 
\[q | (b + i_jc) - (b+i_{j'}c) = (i_j - i_{j'})c\]
Since $q> |(i_j - i_{j'})|$ and $q \nmid c$, this can only happen if $i_j = i_{j'}$. This must hold for every pair, and thus $i_j = i^*$ for all $j$. In other words, $qp_1p_2\ldots p_k | (b+i^*c)$, which implies that ?
\end{proof}

In this subsection we give an equivalent characterization that may be easier to work with. 

We start with a useful lemma:
\begin{lemma}
Let $A = \{b + ic: 1 \le i \le \ell\}$ for integers $b \ge 0, c > 0$. Let \[A' = \{c^{\log_2 \ell}(b + ic): 1 \le i \le \ell\}.\] Then $A'$ satisfies the $n$-divisor property if $A$ does, and every positive integer at most $\ell$ divides some element of $A'$.
\end{lemma}
\begin{proof}
Let $m \le \ell$ be a positive integer. The key observation is that $\gcd(c^{\log_2 n}c, m)$ divides $\gcd(c^{\log_2 n}b, m)$. To see this, set $r_0 = \gcd(c, m)$ and observe that the prime factors of $r_0$ each may occur at most $\log m \le \log \ell$ times in $m$. Write $m = rm'$ where $r'$ consists of prime factors of $r_0$ in their multiplicity and $\gcd(m', c) = 1$. Then $\gcd(c^{\log_2 n}c, m) = r$, and $\gcd(c^{\log_2 n}b, m) = r\cdot\gcd(b, m')$.

By standard facts about modular equations, then,  there is a solution to \[i(c^{\log_2 n}c) \equiv (-c^{\log_2 n}b) \pmod{m},\] with $1 \le i \le \ell$, as required.
\end{proof}

\begin{lemma}
Suppose that the $n$-Divisor property holds for the set $A = \{b + ic: 1 \le i \le \ell\}$, with $b,c \le \exp(n^{\alpha})$, for $\alpha < 1/2$. Then $\ell \ge 1 - \alpha - o(1)$.
\end{lemma}
\begin{proof}
Let $Q$ be the set of primes between $\sqrt{n}/1000$ and $\sqrt{n}/100$; by the Prime Number Theorem, $|Q| \ge \Omega(\sqrt{n}/\log n)$. Our bound will come from considering those $m \in [n]$ such that $m = pq$ for $p, q \in Q$.

Since $c \le \exp(n^\alpha)$ there are at most $O(n^\alpha) \ll |Q|$ distinct primes in $Q$ that divide $c$, so there must be some $p_0 \in Q$ such that $p_0 \nmid c$. There are $|Q|$ distinct integers of the form $p_0q$ with $q \in Q$, all of which must divide some element of $A$ by the $n$-Divisor property. Define \[I = \{i: \exists q \in Q \mbox{ such that } p_0q | (b+ic)\}.\]
For every $i, j \in I$, we have that $p_0$ divides $(b+ic)$ and $p_0$ divides $(b+jc)$  which implies $p_0$ divides their difference $(j - i)c$; then since $p_0 \nmid c$, it must be that $p_0$ divides $(i-j)$. Thus, the differences $(i - j)$ for $i, j \in I$ must all be divisible by $p_0$, and thus $|I| \le \ell/p_0$.

By the $n$-Divisor property, for every $m$ of the form $m = p_0q$, there must be some $i \in I$ for which $m|(b+ic)$. Since $(b+ic) \le (\ell +1)\exp(n^\alpha)$ we see that at most $O(n^{\alpha} + \log \ell)$ different such $m$'s can divide the same $(b+ic)$. We conclude that 
\[|Q| \le O(n^{\alpha} + \log \ell)|I| \le O(n^{\alpha} + \log \ell)\cdot(\ell/\sqrt{n}),\]
and combining with $|Q| \ge \Omega(\sqrt{n}/\log n)$ we get $n \le O(n^{\alpha+o(1)})\cdot \ell$ as claimed.
\end{proof}


An algebraic fact commonly used by modern factoring algorithm is that over $\F_q$ the polynomial
\[X^{q^n} - X \mod f(X)\]
is the product of factors of $f(X)$ with degree dividing $n$. The polynomial $X^{q^i}$ is also called the \textit{i-th Forbenius power} of $X$. To quickly compute such a Frobenius power modulo $f(X)$,  we utilize the invention due to (Kaltofen, von zur Gathen \& Shoup [vzGS92]): first compute $X^q \mod f(X)$ by repeated squaring, then compose the result with itself modulo $f(X)$ to get

\[(X^q)^q \mod f(X) = X^{q^2} \mod f(X).\]

Therefore, it only takes $O(\log n)$ modular compositions to obtain $X^{q^n} \mod f(X)$ from $X^{q} \mod f(X)$. A crucial fact used is that raising to the $q$-th power commutes with modular composition, which is not true for multiplications in general. We include a lemma here and defer the proof to the appendix. 

\begin{lemma} Over $\F_q[X]$, \[X^{q^{i+j}} \mod f(X) = \left(X^{q^{i}} \mod f(X) \right)\circ \left(X^{q^{j}} \mod f(X)\right) \]
    
\end{lemma}
\begin{proof}
    See Appendix (?)
\end{proof}

In \cite{KU11}, Kedlaya and Umans obtained an algorithm that solved DDF for a degree $n$ polynomial over $\F_q$ using $\widetilde{O}(\log n^{3/2 + o(1)})$ bit operations. To summarize their strategy, they computed polynomials of the form
\[P_m(X) = (X^{q^1} - X)(X^{q^2} - X)\ccdot (X^{q^m} - X) \mod f(X)\]
for some $1 \leq m \leq n$. Using the baby-step, giant-step approach, we obtain $P_m(X)$ by evaluating a degree-$O(\sqrt{m})$ polynomial at $O(\sqrt{m})$ points, at a total cost of $\widetilde{O}(\sqrt{m}\cdot n)$ operations. Since $\gcd(P_m(X), f(X))$  contains all factors of $f(X)$ with degree $m$ or less, by checking if $\gcd(P_m(X), f(X))$ equals $1$, or $f(X)$, or some non-trivial factor $g_m(X)$ of $f(X)$, one can reduce the valid interval $[1,...,m]$ that an irreducible factor lies by $1/2$ every time a $\gcd$ computation is performed. As a result, using binary search on the mid-point of each interval (e.g. $m = n/2$, $m = n/4, m = 3n/4 ...)$, and repeatedly computing $\gcd(P_{m}(X), f(X)/g_{m}(X))$, one can efficiently produce all irreducible factors of $f(X)$.\\

In the algorithm above, we have hoped to compute $X^{q^{m}} - X \mod f(X)$ for small $m$'s (namely, $m \leq n$). To get an improvement of the exponent-$3/2$ algorithm, notice that one such computation only takes $O(\log m \log q)$ operations, which is far below the exponent-$3/2$ allowable budget. If we let $m$ grow, encompassing more divisors of $n$ and therefore reducing the number of intermediary polynomials $X^{q^{m}} - X \mod f(X)$ to compute, we can potentially achieve speed-up of DDT. More formally, we look for a set of integers $A_1,..., A_m$ such that the following property holds:

\SWnote{where should I put the small example?}

\begin{property} \label{divisor}
Let $0 \leq \beta \leq 1$. For some $m = n^\beta$, the set of integers $A_1,.., A_m$ satisfies the divisor property if  $\forall i \in \{1,.., n\}$, $i \mid A_j$ for some $j$.
\end{property}

For $m = n$, we have that $A_1 = 1, ..., A_m = n$ trivially holds the divisor property. For $m = n/2$, it is an easy exercise to check that $A_1 = n/2 + 1, ..., A_{n/2} = n$ also holds the divisor property. \\

Corresponding to Property \ref{divisor}, we present a version of the DISTINCT DEGREE FACTORING problem we'd like to solve:

\begin{problem} \label{poly}
    Given $A_1, ..., A_m$, that satisfies Property \ref{divisor}, compute 
    \[t_1(X) \ccdot t_m(X) \equiv (X^{q^{A_1}} - X) \ccdot (X^{q^{A_m}} - X) \mod f(X)\]

    where each $t_i(X)$ is a product of factors of $f(X)$ with degree dividing $A_i$. 
\end{problem}

Continuing with the example above, the case $m = n$, $A_1 = 1, ..., A_n = n$ corresponds to the exponent-$3/2$ algorithm, and for $m = 1$ and $A_1 = m!$ we can trivially compute\[X^{q^{m!}} - X \mod X\]
using modular composition, which leads to an exponent-$2$ algorithm. Therefore, to get an improvement on DISTINCT DEGREE FACTORING, additional constraints on the magnitude and the structure of these numbers need to be imposed. We capture these constraints as follows:

\begin{conjecture} ($(\alpha,\beta)$-Divisor Conjecture)
\label{gap}  Let $\alpha \in (0, 1/2), \beta \in (0, 1)$. For every $n > n_0$, there exists $S, T \subseteq \ZZ^{+}$ such that the pairwise differences \[\{s-t \mid s \in S, t \in T\}\] satisfies the divisor property. Moreover, the following holds:\\
1. $|S|, |T|$ are of size $n^{\beta/2+o(1)}$;\\
2. $s, t = \exp(n^\alpha)$; \\
3. $\exists c, c' = n^{o(1)}$ such that
\[S = \{S_1 + ... + S_c\}, T = \{T_1 + ... +T_{c'}\}.\]
where $S_i, T_j$ are arithmetic progressions.
\end{conjecture} 

\SWnote{Elaborate } Notice that $\alpha + \beta \geq 1$ by the magnitude of the product of primes $\leq n$. Therefore, $\beta \geq 1/2$. \\

This is the same as saying that the integers in $S$ (or $T$) are vertices in a hypercube of volume $O(n^{\beta/2+o(1)})$. Intuitively, if a set of a integers have such a ``succinct representation''; it is computationally cheap to obtain one from another. Assuming that Conjecture \ref{gap} is true, Theorem \ref{GAP} shows a pathway to solving Problem \ref{poly}. On the other hand, it is not true that we can directly reuse the binary search process from \cite{KU11} to get an improvment on DDF. Several adjustments need to be made and these are captured by Corollary \ref{bs}.\\

We first fix some notations and parameters. For the remainder of this section, $g(X)$ is a polynomial of degree $d \leq n$ which is understood to be a factor of $f(X)$.\\

For an integer $A$ and a polynomial $g(X)$ of degree $d$, we call the polynomial
\[X^{q^{A}} \mod g(X)\]
\textbf{$A$'s remainder modulo $g(X)$}. If the context is clear, we will simply call it \textbf{modulo remainder}. 

\begin{lemma}\label{lem:q-rem} For an integer $A_i$ of magnitude $\exp (n^{\alpha})$ and $g(X)$ of degree $d$, its remainder modulo $g(X)$
\[X^{q^{A_i}} \mod g(X)\]
takes $\widetilde{O}(n^{\alpha}\cdot d)$ operations to compute. 
\end{lemma}
\begin{proof}
   First, we obtain $X^q \mod g(X)$ via repeated squaring: $X \mod g(X), X^2 \mod g(X), X^4 \mod g(X),...$. This takes $O(\log q \cdot d)$ operations over $\F_q[X]/g(X)$. To compute a remainder of the form \[X^{q^{A_i}}\mod g(X),\] notice by the magnitude of $A_i$, it has a binary representation of $M$ bits, for some $M = O(n^\alpha)$. Write $A_i = a_1\cdot 2^{1} + \ccdot + a_M \cdot 2^{M}, a_i \in \{0, 1\}$, we can compute modulo remainder associated to each digit place
\[X^{q^{2}} \mod g(X), ...,  X^{q^{2^{M}}} \mod g(X).\]

It takes a modulo composition of $X^{q} \mod g(X)$ with itself to obtain $X^{q^2} \mod g(X)$; the same is true for obtaining the modulo remainder of any $2^{i}$ from $2^{{i-1}}$. Once each digit place is done, it take another $M$ compositions to assemble $A_i$'s modulo remainder. The Lemma is proved.   
\end{proof}

Now, we define an ordering on the set $S$ (or $T$). As $S$ is a generalized arithmetic progressions of $c$ differences, first define an ordering on the differences: $d_1 \prec d_2 \prec ... \prec d_c$. Notice that the ordering does not have to base on magnitude; it can be symbolic. Then, for $a, b \in S$, let $a \prec b$ iff $a_i \prec b_i$ where $i$ is the smallest $d_i$ where $a$ and $b$ differ. Denote the ordering as $\prec_S$ on $S$ (respectively, $\prec_T$ on $T$). \\ 

\begin{lemma}\label{lem:q-rem-set} 
Let $0 < \alpha < 1/2$. Using the notations from Conjecture \ref{gap}, let $S_0 \subseteq S$ (or $T$) be a subset of size $k$ consisting of consecutive elements. It takes $O((n^{\alpha} + k)\cdot n^{o(1)} \cdot d)$ operations to obtain the set $S_0$'s remainders modulo $g(X)$. 
\end{lemma}

\begin{proof}
A consecutive subset $S_0$ of $S$ inherits the ordering, and is still a generalized arithmetic progression of $\leq c$ differences. Let $s$ be the smallest element of $S_0$. Its remainder modulo $g(X)$ has the form 
\[X^{q^{s}} \mod g(X)= (X^{q^{d_0}}) \circ (X^{q^{i_1 d_1}}) \circ ... \circ (X^{q^{i_c d_c}})\mod g(X)\]
for some initial offset $d_0  = \exp(n^\alpha)$ and  $\prod_{k = 1}^c{i_k} \leq n^{\beta/2 + o(1)}$. By Lemma \ref{lem:q-rem}, it takes $(c+1) \cdot O(n^{\alpha}\cdot d) = O((n^{\alpha + o(1)})\cdot d)$ operations to obtain the modulo remainder associated with all the $d_i$'s; composing each gate for at most $O(\log n)$-many times give us the modulo remainder of $s$. To compute the whole set $S_0$, since all elements are consecutive, one needs $k$ extra compositions from $s$ to get all the modulo remainders. The Lemma follows. \\    
\end{proof}

Let $S$, $T$ be defined as in Conjecture \ref{gap}. We now order the set $\{s - t \mid s \in S, t \in T\}$ with $s \in S$ changing most rapidly, and then $T$. Within $S$ and $T$ the elements are ordered via $\prec_S$ and $\prec_T$. Relabel the set $M = \{A_1, ..., A_m\} = \{s - t \mid s \in S, t \in T\}$, for some $m = n^{\beta + o(1)}$. For an $A_i = s_a - t_u, A_j = s_b - t_v$, we have $A_i \prec_{M} A_j$ iff $s_a \prec_S s_b$ or $s_a = s_b$ and $t_a \prec_T t_b$. Notice that this ordering is also symbolic, and not based on magnitude.\\

For a polynomial $g(X)$ of degree $d$, define the \textbf{interval polynomial modulo $g(X)$} as \[Q(u,v) = \prod_{i = u}^v(X^{q^{A_i}} - X) \mod g(X)\]

which is the product of modulo remainders from $A_u$ to $A_v$. The following lemma shows that the complexity of computing interval polynomials of the same length is the same regardless of starting points. 

\begin{lemma} \SWnote{change this to the equivalent of: ... that satisfies Property 2}\label{interval_poly}
 For $1\leq u \leq v \leq m$, the interval polynomial $Q(u,v) \mod g(X)$ takes $O((n^{\alpha} + n^{\beta/2})\cdot n^{o(1)} \cdot d)$ operations to compute. This solves Problem \ref{poly}. 
\end{lemma}

\begin{proof}
Write $A_u = s_{i_u} - t_{j_u}$ and $A_v = s_{i_v} - t_{j_v}$. As $Q(u,v)$ is the product of consecutive modulo remainders, we first break it up into three pieces: an initial segment 
\[Q_0 = \prod_{k = j_u}^{|T|}(X^{q^{s_{i_u} - t_{k}}}- X)\]
a middle segment 
\[Q_1 = \prod_{k = 1}^{i_v-i_u}\prod_{t \in T}(X^{q^{s_{i_{u}+k}-t}}- X)\]
and a final segment 
\[Q_2 = \prod_{k = 1}^{j_v}(X^{q^{s_{i_v} - t_{k}}}- X)\]
Notice that both $Q_0 \mod g(X)$ and $Q_2 \mod g(X)$ are the modulo remainders of subsets of $T$ shifted by $s_{i_u}$ and $s_{i_v}$. Since $Q_0$ and $Q_2$ have length at most $|T| = n^{\beta/2+o(1)}$, using Lemma \ref{lem:q-rem-set}, we can obtain $Q_0 \mod g(X)$ and $Q_2 \mod g(X)$ in time $O((n^{\alpha} + n^{\beta/2})\cdot n^{o(1)} \cdot d)$. To obtain $Q_1$, we employ a baby-step, giant-step approach:\\

Over the ring $\F_q[X]/g(X)[Z]$, we form the polynomial $P(Z)$
\[P(Z) = \prod_{k=1}^{i_v-i_u}(X^{q^{s_{i_u+k}}} - Z) \mod g(X)\]
and evaluate it at the modulo remainders of the set $T$. Since $P(Z)$ has degree at most $n^{\beta/2 + o(1)}$, forming $P(Z)$ takes time $O((n^{\alpha} + n^{\beta/2})\cdot n^{o(1)} \cdot d)$. Using fast multipoint evaluation \SWnote{add reference}, the cost of evaluating a degree $\leq n^{\beta/2 + o(1)}$ univarite polynomial at $\leq n^{\beta/2 + o(1)}$ points is $\widetilde{O}(n^{\beta/2} \cdot d)$ over $\F_q[X]/g(X)$. This procedure produces a polynomial of the form 
\[\prod_{k = 1}^{i_v-i_u}\prod_{t \in T}(X^{q^{s_{i_{u}+k}}}- X^{q^t}) \mod g(X).\]
By the property of $\F_q$, this is the same as \[\prod_{k = 1}^{i_v-i_u}\prod_{t \in T}(X^{q^{s_{i_{u}+k}-t}}- X)^{q^t} \mod g(X).\] 
Since $g(X)$ (as a factor of $f(X)$) is squarefree, the modulo remainder computes the same polynomial as $Q_1 \mod g(X)$. As a result, $Q_0Q_1Q_2 \mod g(X)$ gives us the desired interval polynomial $Q(u,v) \mod g(X)$. 
\end{proof}

In the following theorem, we refer to the interval polynomial $Q(u,v)$'s \textbf{range} as the set of consecutive integers $\{A_u,..., A_v\}$.
\begin{theorem} \SWnote{Maybe change this into more about Property 3}\label{prop_3}
Let $\gamma = \max(\alpha,\beta/2)$. There exists a probabilistic algorithm for DISTINCT DEGREE FACTORING in time $O(n^{\gamma+1+o(1)})$ over $\F_q[X]$. 
\end{theorem}

\begin{proof}
   We adopt the notation from the previous proof. We start with a preprocessing stage, where all irreducible factors of degree $\leq n^{2\gamma}$ are produced. 
   Using the exponent-1/2 algorithm in \cite{KU11}, we can do this in time $\widetilde{O}(n^{\gamma})$.\\
    
   There are three main components to this algorithm. By Lemma \ref{interval_poly}, the interval polynomials we described satisfy Property 2 of Theorem \ref{recur-split}. So the first component is recursive splitting, as described in Theorem \ref{recur-split}.  So the result of recursive splitting gives us integers $\{(s-t)_1, ..., (s-t)_{b_0}\}$, $b_0 = \widetilde{O}(n^{\beta/2})$ \SWnote{actually this needs some work, because more straightforward is $n^{1-2\gamma}$ but it's not so pleasing} each of magnitude $\exp(n^{\alpha})$, such that their $q$-power remainders modulo $f(X)$, i.e.
\[X^{q^{(s-t)_i}} - X \mod f(X), i = 1,..., b_0\]produce a set 
$\bigg\{k_1(X), ..., k_{b_0}(X)\bigg\}$ for which $\prod_{i = 1}^{b_0} k_i(X) = f(X)$.   We need to show that for each integer $(s-t)_i$ and its modulo remainder $k_i(X)$ of degree $d_i$, we can compute the distinct-degree factorization of $k_i(X)$, i.e., a set of polynomials $f_1, ..., f_n$ such that $k_i(X) = f_1, ..., f_i$ and $f_i$ is the product of degree $i$ irreducible polynomials. This will be the second and the third components of the algorithm. Combined, they satisfy Property $3$ of Theorem \ref{recur-split}, and we therefore get DDF.\\

The goal of the second stage is to reduce the size of each integer $(s-t)_i$ from $\exp(n^{\alpha})$ to $O(n)$ by both removing irrelevant factors its prime factorization, and split down each $k_i(X)$ into irreducible factors of degree $\leq O(n)$.\\

An assumption we make is that for integer $(s-t)_i$, we are supplied with its prime factorization (with repeated primes): $(s-t)_i = p_{i_1}\ccdot p_{i_{a_i}}$, for which $a_i = O(n^{\alpha})$.\\

 The input of the algorithm is an integer $(s-t)_i$, its associated modulo remainder $k_i(X)$ and its prime factor set $S = \{p_{i_1},...,p_{i_{a_i}}\}$.  In one round, the algorithm discards each prime $p_{i,j}$ in $\{p_{i_1},...,p_{i_{a_i}}\}$ independently with probability $\frac{1}{2\log n}$; the ones that are left form a new set $S_1$. Then, the $q$-th power remainder associated with $S_1$ is computed:
\[r_1(X) = X^{q^{\prod_{p \in S_{1}}p}} - X \mod k_i(X)\]

Both procedures run in $O(n^{\alpha}\cdot d_i)$, where $d_i$ is the degree of $k_i(X)$. We claim the probability that $\gcd(r_1(X), f(X)) = g(X)$ is a non-trivial factor of $k_i(X)$ is $\geq 1/2$. This is because the prime factors of the degree of $g(X)$ lies in the set $S_0$, and there are at most $\log n$ of them. The probability that one such factor get discarded is $\frac{1}{2 \log n}$; therefore, by union bound, the probability that none of them gets discarded is $\geq \log n/(2\log n) = 1/2$.\\

Therefore, $g(X) = 1$ means we have either failed to split a factor, or $k_i(X)$ does not have any factors of degree dividing $S_1$. In either case, we can repeat this round by re-picking $S_1$ for $(\log n)^2$ times to reduce the failure probability to $\leq (1/2)^{(\log n)^2} = 1/n^{\log n}$. Therefore, with a running time of $O(n^{\alpha}\cdot d_i)$, we will see that $\gcd(r_1(X), f(X)) = g(X)$ is a nontrivial factor of $k_i(X)$. On the other hand, for every irreducible factor of $k_i(X)$, the probability that it is preserved is $\geq 1/2$; so in expectation, $g(X)$ contains more than $1/2$ of $k_i(X)$'s factors. We can move on to the next round with two subproblems, $\{S_1, g(X)\}$ and $\{S, f(X)/g(X)\}.  \\

On the other hand, If $g(X) = k_i(X)$, we move on to the next round with the new factor set $S_1$; if $g(X)$ is a nontrivial factor of $k_i(X)$, we can move on to the next round with $S_1$ and $g(X)$.\\ 

Since every successful round reduce the expected size of $S_0(X)$ by a factor of $(1-\frac{1}{2\log n})$, after $<2 \log n \ln n$ rounds we would have reduced the size of the set to $O(\log n)$ \SWnote{maybe add short reasoning}. When this happens, we proclaim the probabilistic stage as done.\\

In terms of running time, the cost of computing the $q$-th power remainder associated with any $S_i$ modulo a $k_i(X)$ with degree $d_i$ is $\widetilde{O}(n^{\alpha} \cdot d_i)$. Since there are at most $O((\log n)^2)$ levels of splitting, and on each level, the degree of the modulo remainders always sum up to $n$, so the total work is $\sum (\widetilde{O}(n^{\alpha}) \cdot d_i) = \widetilde{O}(n^{\alpha}) \cdot n$ run at most $O((\log n)^2)$ times. Therefore, we finish this probabilistic stage with running time $\widetilde{O}(n^{\alpha+1})$. The outcome of this stage is a set of integers $\{s_1, ..., s_c\}$ for some $c = \widetilde{O}(n^{\beta/2})$, each $s_i$ has a prime factorization of at most $O(\log n)$ primes, as well as their $q$-th power remainder $\{k'_1(X),..., k'_c(X)\}.$\\

In the final deterministic stage, the goal is to produce a distinct-degree factorization of $f(X)$: $f_1, f_2,..., f_n$ where $f = f_1f_2 \ccdot f_n$ and each $f_i$ is either 1 or the product of degree-$i$ irreducible polynomials. This is a probabilistic solution to Problem 1.1. \\

For each $s_i$, its modulo remainder $k_i'(X)$ and its prime factor set $S = \{p_1,...,p_{a}\}$ for some $a = O(\log n)$ , compute 
\[r'(X) = X^{q^{\prod_{p \in S \setminus p_k} p}} - X \mod k_i'(X)\]

for each $p_k \in \{p_1,...,p_{a}\}$. That is, in each round we remove one prime factor $p$ from the set $S$ and see if $S \setminus p$ produce a non-trivial split. If the remainder $\gcd(r'(X), k'_i(X))$ is $1$ for every $p_k$, then $k'_i(X)$ is exactly the product of irreducible factors of degree $\prod_{p \in S}p$, and we can output this answer. If $\gcd(r'(X), k'_i(X)) = k'_i(X)$ for some $S \setminus p_k$, we can discard $p_k$ and move on to the next round. Finally, if $\gcd(r'(X), k'_i(X))$ a nontrivial split of the factors, we can output $k'_i(X)/r(X)$ as an irreducible factor of degree $p_k$, and continue the search on $r(X)$ with the set of primes $S \setminus p_k$. For each $k'_i(X)$, this procedure runs for at most $(\log n)^2$ rounds; So the cost of splitting $k'_i(X)$ into distinct degree irreducible factors is $\widetilde{O}(d_i') \cdot \text{polylog}(n)$ for $\deg(k'_i) = d_i'$. Since $\sum \deg(k'_i) = n$, we have that the running time for step 3 is...
\end{proof}

\begin{proposition} (\textcolor{red}{Need fixing})
    $\alpha = \beta/2 = 1/3$ achieves the optimal running time of solving Problem \ref{polyfactor} using the baby-step, giant-step scheme \SWnote{Needs renaming}.
\end{proposition}
\begin{proof}
Optimally we have $\gamma + \alpha = 1$. If $\widetilde{O}(n^{\gamma})$ integers of size $\exp(\widetilde{O}(n^{\alpha}))$ needs to be divisible by $1,...,n$, then they are divisible by the product of all primes from $1,.., n$, which have size ? 
\end{proof}

\newpage
\SWnote{sections below needs fixing}
\section{Related conjectures and the information-theoretic barrier}

We now present a few variations of the conjecture used in Theorem \ref{GAP}. We also give an information-theoretic argument on why a randomly picked set of integers is unlikely to have the property we need. 

Throughout the following, parameters $\alpha$ and $\beta$ have $\alpha + \beta = 1$. As proved in ??, $\alpha = \widetilde{O}(n^{1/3})$ and $\beta = \widetilde{O}(n^{2/3})$ gives the best case running time in our current scheme. 
\begin{conjecture} (Arithmetic Progression)
For every $n$, there exists an integer $B$ of size $\widetilde{O}(exp(n^{\alpha}))$, such that for some $m = \widetilde{O}(n^{\beta})$, we have
$$1,..., n \mid \prod_{i = 0}^{m}(B + i),$$
\end{conjecture}

\begin{remark}
Using Chinese Remainder theorem, this has an equivalent interpretation: Given a set of congruences $(a_1, ..., a_k) \mod (p_1, ..., p_k)$ where $(p_1, ..., p_k)$ are all the primes between $m+1$ and $n$, is there a solution $B \mod \prod {p_i}$ that is of size $O(exp(n^{\alpha}))$?

Another formulation is this:
 Let $p_1 < p_2 < ...< p_l$ denote the sequence of primes larger than $m+1$ and smaller than $n$. There exists an integer $B$ with the following residue classes:

$$B \equiv (P_1-1 \; \makebox{mod}\ P_1)$$
$$B \equiv(P_2-2 \; \makebox{mod}\ P_2)$$
$$\ccdot$$
$$B \equiv(P_k -k \; \makebox{mod}\ P_k)$$

where $l = O(n^{2/3})$ as above, and each $P_i = \prod_{j=1}^{O(n^{1/3})} p_{ij}$ is a product of at most $O(n^{1/3})$ primes (namely those that divide each $B+i)$. CRT guarantees a unique solution mod $P_1\ccdot P_k$, but we want $B = O(n^{n^{1/3}})$. It seems unlikely that the unique solution mod $P_1\ccdot P_k$ is of the same order of magnitude as the residues of all equivalence classes. 

However, it is possible that the final solution constructed via CRT is a bit larger, and any integer with size $<O(n^{n^{1/2}})$ gives us an improvement of the polynomial factoring algorithm. It is in our interest to obtain any non-trivial bound of the magnitude of such an integer. 
    
\end{remark}

\begin{conjecture} (Arithmetic Progression)
For every $n$, there exists a number $B$ of size $\Tilde{O}(n^{n^{1/3}})$, such that for some $k = \Tilde{O}(n^{2/3})$,
$$1,..., n \mid \prod_{i = 0}^{k}(B + iA)$$
\end{conjecture}

\begin{conjecture} (Generalized Arithmetic Progression) There exist $c = \widetilde{O}(\log n)$ integers $A_1, ..., A_c$, each of size $\widetilde{O}(n^{n^{1/3}})$, and  $m_1, ..., m_c \geq 0$ such that 
$$1,..., n \mid \prod_{i_1, ..., i_c = 1}^{m_1, ..., m_c} (B + i_1 A_1 + ... + i_c A_c)$$
 with the constraint that $\prod_{i = 1}^c m_i = \widetilde{O}(n^{2/3})$. In other words, all indices $(i_1, ..., i_c)$ are contained in a hypercube of volume $\widetilde{O}(n^{2/3})$. 
\end{conjecture}

\begin{remark}
Just like the above, any $B, A_1,..., A_c$ of magnitude $< \widetilde{O}(n^{n^{1/2}})$ and $\prod_{i = 1}^c m_i < \widetilde{O}(n)$ will lead to a non-trivial improvement in FACTOR DEGREE. 
\end{remark}

The following two conjectures are more general, but we don't know how to turn them into algorithms. 

\begin{conjecture} (Most general) For every $n$, there exists a set $S$ of size $\Tilde{O}(n^{1/3})$, with each $s \in S$ of size $\Tilde{O}(n^{n^{1/3}})$, such that 
$$1,..., n \mid \prod_{s,t \in S}(s-t)$$
\end{conjecture}

\begin{conjecture} (Residue class formulation)
There exists a set $S \subseteq \Z^{\Tilde{O}(n^{n^{1/3}})}$, $k = |S| = \Tilde{O}(n^{1/3})$, a set $T$, $|T| < k$, such that for every $j \in [1,...,n]$ and every $s \in S$, we have $s \mod j \in T$. 
\end{conjecture}

To disprove Conjecture 6 and 7, one might look into upper bound results on the number of distinct prime factors dividing arithmetic progressions of short lengths. Related papers [...], but this seems to be an open problem.\\

Here is a characterization on the set of integers that is usable to us; that is, they need to have a ``succinct representation".

\begin{proposition} (Information-theoretic barrier) Let $A_1, ..., A_m$ be integers such that $$1,..., n \mid A_1, ..., A_m$$
Suppose each $A_i$ can be written as
$$A_i =  B_0 + i_1B_1 + ... + i_cB_c$$
for $i_1 \leq m_1,..., i_c \leq m_c$ and $\prod_{i = 1} m_c\leq m$, then the set $A_1, ..., A_m$ has a representation in $$\sum_{i = 0}^c  \log(B_i) + \sum_{i= 0}^c \log m_i$$
bits. 
\end{proposition}

\begin{remark}
    In general, we need $\sum_{i = 1}^m\log A_{i}$ bits to represent the set of integers $A_1, ..., A_m$. 
\end{remark}

\section{Integer Factorization}

\begin{problem}\label{intfactor}
Given a composite, squarefree integer $N$, output the prime factors of $N$. 
\end{problem}

\begin{problem}\label{intsub}
Given a composite, squarefree integer $N$ and an integer $m \leq \sqrt{N}$, compute the remainder 
$$m! \mod N$$
    
\end{problem}

\begin{proposition}
 The existence of $A, B$ and the set $S$ outlined in Proposition \ref{polyfactor} gives us an algorithm to solve Problem \ref{intsub} in time $\Tilde{O}(N^{1/6})$. 
 \end{proposition}

\begin{proof}
The $A, B$ associated with $m$ have magnitude $\Tilde{O}(m^{m^{1/3}})$,so we can write them in $m$-ary representations $a_k m^k + a_{k-1} m^{k-1} + ... + a_1 m + a_0$ with $k = O(m^{1/3})$ bits. Computing each bit via repeated squaring mod $N$ takes $\Tilde{O}(\log N)$, so in total computing $A$ and $B$ takes at most $O(m^{1/3} \log^k N)$ operations. Then, constructing the product
$$\prod_{j=0}^{O(m^{2/3})}(B + jA) \mod N$$using the same method as  \ref{polyp} takes $O(m^{1/3})$ operations. Plugging in $m = \sqrt{N}$ gives us the result. \end{proof}

\begin{proposition}
If Problem \ref{intsub} can be solved in xxx, then \ref{intfactor} can be solved in xxx. 
\end{proposition}

 \begin{proof}
We use a binary search strategy to output all factors of $N$, similar to the polynomial factoring problem.
    
    We use $g(m)$ to denote the subroutine of computing $m! \equiv R \mod N$, then taking $\gcd(R, N)$. Every $m$ for which $1 < g(m) < N$ results in a successful split of the factors of $N$. Otherwise, depending on whether $g(m) = 1$ or $N$, we eliminate the corresponding intervals as they contain no factors of $N$. 
    
    Since $N$ can have at most $\log N$ prime factors, and at most one of them is $> \sqrt{N}$, we start off with $m = \sqrt{N}$, compute $g(\sqrt{N})$; then $g(\sqrt{N}/2)$; then $g(\sqrt{N}/4)$ or $g(3\sqrt{N}/4)$....

    We now analyze the operation count of this recursive algorithm. Notice that we never set m larger than $\sqrt{N}$ throughout the entire algorithm, so we will pessimistically assume it is always $\sqrt{N}$  to simplify the analysis.
    
    Since $n$ is composite, at most one prime factor of $n$ is $\geq \sqrt{n}$. We first set $m = \sqrt{n}$ and use $A, B$ and the arithmetic progression associated with this $m$.
 
Call the result $d$. if $d = 0$ then all prime factors are less than $\sqrt{n}$. Now the available range of prime factors are $1,..., \sqrt{n}$, so we set $m = \sqrt{n}/2$ to be the midpoint of this range, and repeat the process. If $d \neq 0$, we know that both $d \leq \sqrt{n}$ and $n/d > \sqrt{n}$ are factors of $n$. In fact $n/d$ is prime, since $n$ cannot have more than two prime factors $> \sqrt{n}$. So we repeat the process on $m = d/2$ and update the range of primes to $1, ..., d$. 

If for any midpoint $m$ we get a $1 \mod n$ as result, we know that $n$ has no factors in the first half of the range, so we shrink the range to the second half and update the midpoint. Since $n$ has at most $\log (n)$ unique prime factors, we need to repeat the algorithm for at most $\log (n)$ times, giving us $\Tilde{O}(n^{1/6})$ running time in total. 
\end{proof}

In the polynomial factorization case, we described how $\Tilde{O}(n^{3/4})$ is the optimal running time within the baby-step, giant step scheme using the given arithmetic progression. It turns out that in the integer's case...

\begin{proposition}
    Suppose there exists integers $A$, $B$ of magnitude $\Tilde{O}(m^{m^{1/2}})$ s.t.
    $$1,..., m \mid \prod_{j=0}^{k} (B + jA)$$ where $k = \Tilde{O}(n^{1/2})$. Moreover, if both $A, B$ are both products of consecutive integers, then  can be solved in time $\Tilde{O}(n^{1/8})$. 
\end{proposition}

-Can add some observations about the chinese remainder theorem when solving either the stronger conjecture or the weaker one\\

\bibliographystyle{alpha}
\bibliography{ref}

\newpage
Let's consider the (adversarial) case where $X^{q^{S}} - X \mod f(X) = f(X)$, $S$ is of size $\exp(\widetilde{O}(n^{1/3}))$ and $S = p_1\cdot \cdot \cdot p_c$ is its factorization of primes $(c = \widetilde{O}(n^{1/3}))$. Suppose $f(X)$ consists of one irreducible factor of degree $n =p_{\alpha,1} \cdot \cdot \cdot p_{\alpha, \log n}$. \\

To reduce the size of $S$ to $\widetilde{O}(\log n)$, at each round, we keep each $p_i$ with probability $q = 1 - \frac{1}{10 \log n}$, and we do this for many rounds. After 1 round, the expected number of elements is $c \cdot q$. By union bound, the probability that we exclude any factor of $n$ is $\leq \frac{\log n}{10 \log n} = \frac{1}{10}$, so the probability that we included all factors of $n$ is $\geq \frac{9}{10}$.\\

After each round, we break each interval proportionally into a $1 - \frac{1}{10 \log n}$ segment and a $\frac{1}{10 \log n}$ segment, but we have to make sure that after $l$ rounds, the biggest segment is of expected size $\leq \widetilde{O}(\log n)$. $l = \widetilde{O}((\log n)^2)$ will work. As an example for $l = 10 \log n \ln n,$
\[c \cdot \left(1 - \frac{1}{10 \log n}\right)^{10 \log n \cdot \ln n} = c \cdot \left(e^{-1 + o(\frac{1}{(\log n)^2})}\right)^{\ln n} = O(1)\]

But the probability of success is calculated as the probability we include all factors of $n$ at each round, independently. So it is $\geq (\frac{9}{10})^{\widetilde{O}((\log n)^2)} = \frac{1}{n^{\widetilde{O}(\log n)}}.$

\end{document}

\subsection*{Old text}

We first fix some notations. For the remainder of this section, $h(X)$ is a polynomial of degree $d \leq n$ which is understood to be a factor of $f(X)$.\\

For an integer $A_i$ and a polynomial $h(X)$ of degree $d$, we call the polynomial
\[X^{q^{A_i}} \mod h(X)\]
\textbf{$A_i$'s remainder modulo $h(X)$}. If the context is clear, we will simply call it \textbf{modulo remainder}. 

\begin{lemma}\label{lem:q-rem} For an integer $A_i$ of magnitude at most $\exp (n^{\alpha})$ and $h(X)$ of degree $d$, its remainder modulo $h(X)$
\[X^{q^{A_i}} \mod g(X)\]
takes $\widetilde{O}(n^{\alpha}\cdot d)\cdot \log q$ operations to compute. 
\end{lemma}
\begin{proof}
   First, we obtain $X^q \mod g(X)$ via repeated squaring: $X \mod h(X), X^2 \mod h(X), X^4 \mod h(X),...$. This takes $O(\log q \cdot d)$ operations over $\F_q[X]/h(X)$. To compute a remainder of the form \[X^{q^{A_i}}\mod h(X),\] notice by the magnitude of $A_i$, it has a binary representation of $M$ bits, for some $M = O(n^\alpha)$. Write $A_i = a_1\cdot 2^{1} + \ccdot + a_M \cdot 2^{M}, a_i \in \{0, 1\}$, we can compute modulo remainder associated to each digit place
\[X^{q^{2}} \mod h(X), ...,  X^{q^{2^{M}}} \mod h(X).\]

It takes one modulo composition of $X^{q} \mod h(X)$ with itself to obtain $X^{q^2} \mod h(X)$; the same is true for obtaining the modulo remainder of any $2^{i}$ from $2^{{i-1}}$. Once each digit place is done, it take another $M$ compositions to assemble $A_i$'s modulo remainder. The Lemma is proved.   
\end{proof}

Now, we define an ordering on the set $S$ (or $T$, both as in Conjecture \ref{gap}). As $S$ is a generalized arithmetic progressions of $c$ differences, first define an ordering on the differences: $d_1 \prec d_2 \prec ... \prec d_c$. Then, for $a, b \in S$, let $a \prec b$ iff $a_i \prec b_i$ where $i$ is the smallest $d_i$ where $a$ and $b$ differ. Denote the ordering as $\prec_S$ on $S$ (respectively, $\prec_T$ on $T$). \\ 

\begin{lemma}\label{lem:q-rem-set} 
Let $S_0 \subseteq S$ (or $T$) be a subset of size $m$ consisting of consecutive elements. It takes $O((n^{\alpha} + m)\cdot n^{o(1)} \cdot d)$ \SWnote{missing $O(\log q)$} operations to obtain the set $S_0$'s remainders modulo $h(X)$. 
\end{lemma}

\begin{proof}
A consecutive subset $S_0$ of $S$ inherits the ordering, and is still a generalized arithmetic progression of $\leq c$ differences. Let $s$ be the smallest element of $S_0$. Its remainder modulo $h(X)$ has the form 
\[X^{q^{s}} \mod h(X)= (X^{q^{d_0}}) \circ (X^{q^{i_1 d_1}}) \circ ... \circ (X^{q^{i_c d_c}})\mod h(X)\]
for some initial offset $d_0  = \exp(n^\alpha)$ and  $\prod_{j = 1}^c{i_j} \leq n^{\beta + o(1)}$. By Lemma \ref{lem:q-rem}, it takes $(c+1) \cdot O(n^{\alpha}\cdot d) = O((n^{\alpha + o(1)})\cdot d)$ operations to obtain the modulo remainder associated with all the $d_i$'s; composing each gate for at most $O(\log n)$-many times give us the modulo remainder of $s$. To compute the whole set $S_0$, since all elements are consecutive, one needs $m$ extra modular compositions from $s$ to get all the modulo remainders. The Lemma follows. 
\end{proof}

Let $S$, $T$ be defined as in Conjecture \ref{gap}. We now define an ordering $\prec_A$ on the set $A = \{s - t \mid s \in S, t \in T\}$ with $s \in S$ changing most rapidly, and then $T$. Within $S$ and $T$ the elements are ordered via $\prec_S$ and $\prec_T$. Relabel the set according to $\prec_A$ as $A = \{A_1, ..., A_k\}$, for some $k(n) = n^{2\beta + o(1)}$. For an $A_i = s_a - t_u, A_j = s_b - t_v$, we have $A_i \prec_{A} A_j$ iff $s_a \prec_S s_b$ or $s_a = s_b$ and $t_a \prec_T t_b$. Notice that this ordering is symbolic, and not based on magnitude.\\

For a polynomial $h(X)$ of degree $d$, define the \textbf{interval polynomial modulo $h(X)$} as \[Q(u,v) = \prod_{i = u}^v(X^{q^{A_i}} - X) \mod h(X)\]

which is the product of modulo remainders from $A_u$ to $A_v$. The following lemma shows the interval polynomials coming from the set $A$ satisfies Property 2. 

\begin{proposition}\label{interval_poly}
 For $1\leq u \leq v \leq m$, the interval polynomial $Q(u,v) \mod g(X)$ takes $O((n^{\alpha} + n^{\beta})\cdot n^{o(1)} \cdot d)$ operations to compute. Set $\beta = \max(\alpha, \beta) + o(1)$, this satisfies Property 2. 
\end{proposition}

\begin{proof}
Write $A_u = s_{i_u} - t_{j_u}$ and $A_v = s_{i_v} - t_{j_v}$. As $Q(u,v)$ is the product of consecutive modulo remainders, we first break it up into three pieces: an initial segment 
\[Q_0 = \prod_{k = j_u}^{|T|}(X^{q^{s_{i_u} - t_{k}}}- X)\]
a middle segment 
\[Q_1 = \prod_{k = 1}^{i_v-i_u}\prod_{t \in T}(X^{q^{s_{i_{u}+k}-t}}- X)\]
and a final segment 
\[Q_2 = \prod_{k = 1}^{j_v}(X^{q^{s_{i_v} - t_{k}}}- X)\]
Notice that both $Q_0 \mod h(X)$ and $Q_2 \mod h(X)$ are the modulo remainders of subsets of $T$ shifted by $s_{i_u}$ and $s_{i_v}$. Since $Q_0$ and $Q_2$ have length at most $|T| = n^{\beta+o(1)}$, using Lemma \ref{lem:q-rem-set}, we can obtain $Q_0 \mod h(X)$ and $Q_2 \mod h(X)$ in time $O((n^{\alpha} + n^{\beta})\cdot n^{o(1)} \cdot d)$. To obtain $Q_1$, we employ a baby-step, giant-step approach:\\

Over the ring $\F_q[X]/h(X)[Z]$, we form the polynomial $P(Z)$
\[P(Z) = \prod_{k=1}^{i_v-i_u}(X^{q^{s_{i_u+k}}} - Z) \mod h(X)\]
and evaluate it at the modulo remainders of the set $T$. Since $P(Z)$ has degree at most $n^{\beta + o(1)}$, forming $P(Z)$ takes time $O((n^{\alpha} + n^{\beta})\cdot n^{o(1)} \cdot d)$. Using fast multipoint evaluation \SWnote{add reference}, the cost of evaluating a degree $\leq n^{\beta + o(1)}$ univarite polynomial at $\leq n^{\beta + o(1)}$ points is $\widetilde{O}(n^{\beta + o(1)} \cdot d)$ over $\F_q[X]/h(X)$. This procedure produces a polynomial of the form 
\[\prod_{k = 1}^{i_v-i_u}\prod_{t \in T}(X^{q^{s_{i_{u}+k}}}- X^{q^t}) \mod h(X).\]
By the property of $\F_q$, this is the same as \[\prod_{k = 1}^{i_v-i_u}\prod_{t \in T}(X^{q^{s_{i_{u}+k}-t}}- X)^{q^t} \mod h(X).\] 
Since $g(X)$ (as a factor of $f(X)$) is squarefree, the modulo remainder computes the same polynomial as $Q_1 \mod h(X)$. As a result, $Q_0Q_1Q_2 \mod h(X)$ gives us the desired interval polynomial $Q(u,v) \mod h(X)$. 
\end{proof}

\subsection{The information-theoretic barrier}

{\bf CHRIS SAYS: first, this belongs in the section studying the arithmetic progression version since it doesn't discuss the difference sets $S, T$. Once it is there, we might as well keep it simple and talk about a single (not generalized) arithmetic progression. I don't think this part is complete yet. The most significant missing part is that I don't know how one recovers $A_i$ from $A_i'$. You need to specify which factor. Specifying a generic factor costs $O(n^\alpha)$ anyway so there is no contradiction. This needs more thought and we can only include this part if it is fully complete and correct. One comment is that we should probably say "black-box transformation" rather than "procedure" to connect what we are trying to get at to other well-studied scenarios in other areas of complexity.}
Given a set of $m$ integers satisfying the n-divisor property, one may hope for a procedure to output a generalized arithmetic progression containing these integers as divisors, therefore making them useful for our factoring algorithms. We capture the procedure as below:

\begin{procedure}\label{proc:info-theo}
Given integers $A_1,...,A_m$ with magnitude at most $n^{\alpha}$, there is a procedure to output integers $A'_1,...,A'_m$ of magnitude at most $n^{\alpha}$, such that for every $A_i$ there exists an $A'_j$ for which $A_i \mid A'_j$. Further, $A'_1,...,A'_m$ is a generalized arithmetic progression of at most $n^{o(1)}$ differences. 
\end{procedure}

As an example, let $c = \lceil n/m \rceil$ and set $A_1 = \lceil n^{c}\rceil!, A_2 = \prod_{i = 1}^c(\lceil n^{c}\rceil + i)$,.... This partitions the integers from $1$ to $n$ into roughly $m$ groups of $\lceil n/m \rceil$ integers each. The set $A_1,...,A_m$ has magnitude less than $n^{n/m}$. What we wish for is a generic procedure like \ref{proc:info-theo} to work on such inputs. The following theorem says that this is not possible, because a generalized arithmetic progression contains too little information.

\begin{theorem}
There does not exist a procedure satisfying the requirements of \ref{proc:info-theo} that works for every set of $A_1,...,A_m$ with the $n$-divisor property. 
\end{theorem}
\begin{proof}
For the sake of contradiction, suppose Procedure \ref{proc:info-theo} exists. Let $A$ denote the set of integers that is the output of this procedure. This means we can write the set as
\[A =  B_1 + ... + B_c\]
for some $c \leq n^{o(1)}$, and each $B_i$ is an arithmetic progression. We can write each $B_i$ as
\[B_i = \{a_i + jb_i: j = 0, 1,.., d_i\}\]
where $\prod_{i = 1}^c d_i = m$. Information-theoretically, to represent the set $A$, we need to specify each $a_i, b_i$, and the length of the arithmetic progression $d_i$, for $i = 1,.., c$. Further, since $A$ satisfies the $n$-divisor property, for each $i \in [n]$ we need to specify which out of the $m$ integers in $A$ is $i$ a divisor of. In total, this allows us to represent $A$ in \[\sum_{i = 1}^c  \log(a_i) + \log(b_i)+ \log m_i \cdot n= O(n^{\alpha+o(1)}) + \log m \cdot n\]
bits. This amount of information is enough for another procedure to take the set $A$ and produce the original set of integers $A_1,...,A_m$.

On the other hand, there exists $A_1,...,A_m$ satisfying the $n$-divisor property that requires $m \cdot n^{\alpha}$ bits to represent (e.g. each $A_i$ containing as divisor a $n^{\alpha}$-bit prime). Since $O(n^{\alpha+o(1)}) + \log m \cdot n << m \cdot n^{\alpha}$ for $\log n << m$, we cannot get back the input from the generalized arithmetic progression that we produced. As a result, such a procedure does not exist. 
\end{proof}

However, this theorem does not rule out the possibility that a procedure might exist for specific inputs satisfying the $n$-divisor property.

%% file: preamble.tex
\usepackage[utf8]{inputenc}
\usepackage[T1]{fontenc}
\usepackage{lmodern}
\usepackage{amsfonts,amsmath,amssymb,amsthm}
\usepackage[normalem]{ulem}
\usepackage{mathtools}
\usepackage{fullpage}
\usepackage{xcolor}
\usepackage{setspace} 
\definecolor{maincolor}{HTML}{BA0C2F}
\usepackage[normalem]{ulem}
\usepackage[colorlinks,allcolors=maincolor]{hyperref}
\usepackage[capitalize]{cleveref}
\usepackage{tikz-cd}
\usepackage{enumerate}
\usepackage{algorithm}
\usepackage{algpseudocode}
\makeatletter
\def\algbackskip{\hskip-\ALG@thistlm}
\makeatother

\algdef{SE}[SUBALG]{Indent}{EndIndent}{}{\algorithmicend\ }%
\algtext*{Indent}
\algtext*{EndIndent}

\newcommand{\ZZ}{\mathbb{Z}}

\newcommand{\RR}{\mathbb{R}}
\newcommand{\F}{\mathbb{F}}

\newcommand{\poly}{\mathrm{poly}}

\newcommand{\ccdot}{\cdot\cdot\cdot}

\newtheorem{theorem}{Theorem}[section]

\newtheorem{lemma}[theorem]{Lemma}

\newtheorem{problem}[theorem]{Problem}
\newtheorem{definition}[theorem]{Definition}
\newtheorem{proposition}[theorem]{Proposition}
\newtheorem{property}[theorem]{Property}
\newtheorem{conjecture}[theorem]{Conjecture}
\newtheorem{procedure}[theorem]{Procedure}

\theoremstyle{remark}
\newtheorem{remark}{Remark}
\newtheorem*{remark*}{Remark}

\newif\ifmynotes
\mynotestrue
\newcommand{\SWnote}[1]{\ifmynotes\textcolor{blue}{[Siki: #1]}\else\fi}

%% file: distinct_degree_factor.bbl
\begin{thebibliography}{vzGG13}

\bibitem[Ber70]{Berlekamp_1970}
E.~R. Berlekamp.
\newblock Factoring polynomials over large finite fields.
\newblock {\em Mathematics of Computation}, 24(111):713, Jul 1970.

\bibitem[CZ81]{Cantor_Zassenhaus_1981}
David~G. Cantor and Hans Zassenhaus.
\newblock A new algorithm for factoring polynomials over finite fields.
\newblock {\em Mathematics of Computation}, 36(154):587, Apr 1981.

\bibitem[DL80]{Dobkin_Lipton_1980}
David Dobkin and Richard~J. Lipton.
\newblock Addition chain methods for the evaluation of specific polynomials.
\newblock {\em SIAM Journal on Computing}, 9(1):121–125, Feb 1980.

\bibitem[DLS81]{Downey_Leong_Sethi_1981}
Peter Downey, Benton Leong, and Ravi Sethi.
\newblock Computing sequences with addition chains.
\newblock {\em SIAM Journal on Computing}, 10(3):638–646, Aug 1981.

\bibitem[GNU16]{Guo_Narayanan_Umans}
Zeyu Guo, Anand~Kumar Narayanan, and Chris Umans.
\newblock Algebraic problems equivalent to beating exponent 3/2 for polynomial factorization over finite fields.
\newblock {\em 41st International Symposium on Mathematical Foundations of Computer Science}, pages 47:1--47:14, 2016.

\bibitem[Har21]{Harvey_2021}
David Harvey.
\newblock An exponent one-fifth algorithm for deterministic integer factorisation.
\newblock {\em Mathematics of Computation}, 90(332):2937–2950, Jun 2021.

\bibitem[Kal03]{Kaltofen_2003}
Erich Kaltofen.
\newblock Polynomial factorization.
\newblock {\em Proceedings of the 2003 international symposium on Symbolic and algebraic computation}, page 3–4, Aug 2003.

\bibitem[Knu97]{Knuth_1997}
Donald~E. Knuth.
\newblock {\em The Art of Computer Programming}, volume~2.
\newblock Addison-Wesley, 3 edition, 1997.

\bibitem[KS98]{Kaltofen_Shoup_1998}
Erich Kaltofen and Victor Shoup.
\newblock Subquadratic-time factoring of polynomials over finite fields.
\newblock {\em Mathematics of Computation}, 67(223):1179–1197, 1998.

\bibitem[KU11]{KU11}
Kiran~S. Kedlaya and Christopher Umans.
\newblock Fast polynomial factorization and modular composition.
\newblock {\em SIAM Journal on Computing}, 40(6):1767–1802, Jan 2011.

\bibitem[Pol74]{Pollard_1974}
J.~M. Pollard.
\newblock Theorems on factorization and primality testing.
\newblock {\em Mathematical Proceedings of the Cambridge Philosophical Society}, 76(3):521–528, Nov 1974.

\bibitem[Str77]{Strassen_1977}
V.~Strassen.
\newblock Einige resultate uber berechnungskomplexitat.
\newblock {\em Jber. Deutsch. Math.-Verein.}, 1(78):1–8, 1977.

\bibitem[vzG06]{von_zur_Gathen_2006}
Joachim von~zur Gathen.
\newblock Who was who in polynomial factorization.
\newblock {\em Proceedings of the 2006 international symposium on Symbolic and algebraic computation}, page~2, Jul 2006.

\bibitem[vzGG13]{Von_zur_Gathen_Gerhard_2013}
Joachim von~zur Gathen and Jürgen Gerhard.
\newblock {\em Modern computer algebra}.
\newblock Cambridge University Press, 2013.

\bibitem[vzGP01]{von_zur_Gathen_Panario_2001}
Joachim von~zur Gathen and Daniel Panario.
\newblock Factoring polynomials over finite fields: A survey.
\newblock {\em Journal of Symbolic Computation}, 31(1–2):3–17, Jan 2001.

\end{thebibliography}
